\documentclass[jsl]{asl}

\usepackage{xspace}

\usepackage{url}

\usepackage{latexsym}
\usepackage{amssymb}
\usepackage{xspace}
\usepackage{color}
\usepackage[dvips]{epsfig}

\DeclareMathSizes{10}{9}{6}{5}

\newcommand{\fullguarded}{\textup{FullGuarded}}

\newcommand{\canoninst}{\textup{CanonInst}}
\newcommand{\subinst}{\subseteq^w}

\newcommand{\myparagraph}[1]{{\medskip \bf {#1}.}}
\newcommand{\myeat}[1]{}
\newtheorem{theorem}{Theorem}[section]

\newtheorem{lemma}[theorem]{Lemma}

\newtheorem{corollary}[theorem]{Corollary}

\newtheorem{definition}[theorem]{Definition}
\newtheorem{fact}[theorem]{Fact}

\newtheorem{example}[theorem]{Example}

\newcommand{\frakA}{\mathfrak{A}}
\newcommand{\frakB}{\mathfrak{B}}

\newcommand{\limp}{\rightarrow}

\newcommand{\gnnf}{GN-normal form\xspace}

\newcommand\fo{\ensuremath{\textup{FO}}\xspace}

\newcommand\gfo{\ensuremath{\textup{GFO}}\xspace}

\newcommand{\gnra}{\ensuremath{\textup{GN-RA}}\xspace}
\newcommand{\gf}{\gfo}
\newcommand\gnfo{\ensuremath{\textup{GNFO}}\xspace}
\newcommand\gnf{\gnfo}

\newcommand{\gtgd}{\ensuremath{\textup{GTGD}}\xspace}
\newcommand{\fgtgd}{\ensuremath{\textup{FGTGD}}\xspace}

\newcommand{\complexity}[1]{\textsf{\mdseries\upshape #1}}

\newcommand\twoexptime{\complexity{2ExpTime}}

\newcommand{\adom}{\textup{adom}}

\newcommand{\datalogarrow}{:=}

\newcommand{\chase}{\ensuremath{Chase}}
\newcommand{\goal}{\ensuremath{Goal}}

\renewcommand{\vec}[1]{\mathbf{#1}}

\title{Some Model Theory of Guarded Negation}

\author{Vince B\'ar\'any}
\revauthor{B\'ar\'any, Vince}
\thanks{ B\'ar\'any's work done while affiliated with TU Darmstadt.}
\address{Google Inc., Mountain View, CA}
\author{ Michael Benedikt}
\revauthor{Benedikt, Michael}
\address{Department of Computer Science, University of Oxford}
\thanks{Benedikt was supported by EPSRC grant EP/H017690/1}

\author{Balder ten Cate}
\revauthor{ten Cate, Balder}
\thanks{ten Cate was supported by NSF Grants IIS-0905276 
  IIS-1217869.}

\twoaddress{Google Inc., Mountain View, CA}
{Department of Computer Science, UC-Santa Cruz}
{3}

\usepackage{paralist}
\usepackage{booktabs}
\begin{document}
\clubpenalty=10000
\widowpenalty = 10000
  \abovedisplayskip 1ex
  \belowdisplayskip 1ex
  \abovedisplayshortskip 1ex
 \belowdisplayshortskip 1ex


\date{}

\maketitle

\begin{abstract}
The Guarded Negation Fragment ($\gnf$) is a fragment of first-order logic 
that contains all positive existential formulas, can express
the first-order translations of basic modal logic and of many description logics,
along with many sentences  that arise in databases.
It has been shown that the syntax of $\gnf$ is restrictive enough so that 
computational problems such as validity and satisfiability
are still decidable. This suggests that,
in spite of its expressive power, $\gnf$ formulas are amenable to novel
optimizations. In this paper we study the model theory of $\gnf$ formulas.
Our results include effective preservation theorems for $\gnf$, effective
Craig Interpolation and Beth Definability results, and the ability to
express the certain answers of queries with respect
to a large class of $\gnf$ sentences within very restricted logics.

This version of the paper contains streamlined and corrected versions
of results concerning entailment of a conjunctive query from
a set of ground facts and a theory consisting of $\gnf$ sentences
of a special form (``dependencies'').

\end{abstract}


\section{Introduction} \label{sec:intro}

The guarded negation fragment (\gnfo) is a syntactic fragment of first-order logic, 
introduced in \cite{BtCS11icalp} as an extension to the much-studied 
guarded fragment of first-order logic \cite{AvBN98JPL,Gr99JSL}.
Both fragments restrict the use of certain 
syntactic constructs by requiring the presence of \emph{guards}, with the aim 
of taming the language from an algorithmic point of view, with an acceptable compromise 
on expressiveness. The guarded fragment is obtained by requiring all quantification to 
be guarded. This idea has its roots in modal logic and, accordingly, the 
model theory of the resulting fragment has a very similar flavour to that of modal logic. 
The guarded negation fragment is obtained instead by requiring all use of negation to be guarded.
As it turns out, the latter use of guards is more general than the former. Formally, 
every sentence of the guarded 
fragment can be equivalently expressed in the guarded negation fragment~\cite{BtCS15jacm}.
\gnfo also properly contains the positive 
existential fragment of \fo.

\gfo constitutes a rich 
formalism that captures many of the integrity constraint languages and
schema-mapping languages proposed in databases~\cite{dataint,FKMP05},
 and also many of the description logics~\cite{dl} proposed in knowledge representation. 
But \gnfo is more suitable than \gfo for expressing database \emph{queries}; 
that is, mappings from structures to relations. Indeed, as noted above, \gnfo properly 
contains all positive existential formulas.  These are the most common SQL queries,
built up using the basic SELECT FROM WHERE construct and UNION. 

The defining characteristic of \gnfo formulas is that a subformula $\psi(\vec x)$ 
with free variables $\vec x$ can only be negated when used in conjunction with a 
positive literal $\alpha(\vec x, \vec y)$, i.e.~a relational atomic formula or 
an equality atom, containing all free variables of $\psi$, as in 
$$
  \alpha(\vec x, \vec y) \land \lnot \psi(\vec x) \ ,
$$ 
where order and repetition of variables is irrelevant. 
One says that the literal $\alpha(\vec x, \vec y)$ \emph{guards} the negation. 
Unguarded negations $\neg\phi(x)$ of formulas with at most one free variable 
are also supported; this can be seen as a special case
of guarded negation through the use of a vacuous equality guard $x=x$.

It was shown in \cite{BtCS15jacm} that \gnfo possesses a number of desirable 
computational properties. For example, every satisfiable \gnfo formula has a 
finite model (\emph{finite model property}), as well as a, typically infinite, 
model of bounded tree-width (\emph{tree-like model property}).
It follows that satisfiability and entailment (hence, by the finite model property,
satisfiability and entailment in the finite) of \gnfo formulas are decidable.

In \cite{bbo} the implications  of \gnfo for database theory are explored.
For example, an SQL-based syntax for \gnfo is defined, and an analogously constrained 
variant of stratified Datalog is also presented. Several computational problems 
concerning \gnfo formulas (e.g. the ``boundedness  problem'' for a fragment of the 
fixpoint extension of \gnfo) are shown to be decidable.

In this work we investigate  model-theoretic properties of \gnfo.
We first present results showing that \gnfo formulas satisfying specific semantic
properties can be rewritten into restricted syntactic forms.
For example, we show that every \gnfo formula that is preserved under extensions
can be effectively rewritten as an existential \gnf formula.
We give an analogous result for queries preserved under homomorphisms.

Next we consider \gnfo sentences that can also be expressed as a kind of generalized
Horn sentence known in the database community as a tuple-generating dependencies (TGD). 
We provide a syntactic characterization of the \gnfo sentences that are equivalent to 
a finite set of TGDs and give a similar result for sentences in the guarded fragment.

We then turn to model theoretic results concerning explicit and implicit definability.
The Projective Beth Definability theorem states that for any property that is implicitly 
defined by a first-order theory there is a first-order formula that explicitly defines 
the property. We show the analogous result with first-order replaced by \gnfo.
Following ideas of Marx~\cite{Marx07pods} we establish a Craig Interpolation Theorem for \gnfo 
and from this conclude the Projective Beth Definability theorem for \gnfo. This is in 
contrast with the situation for the Guarded Fragment, which does enjoy the simpler Beth 
definability property~\cite{HMO}. Contradicting claims made in earlier work~\cite{Marx07pods} 
we show that Projective Beth fails for the so-called Packed Fragment.

Finally, we  study definability issues related to the ``open world query answering'' problem 
for \gnfo. Open world query answering concerns determining which results of formulas 
are implied by partial information about the underlying structure, in the form of a subset of the 
interpretations of relations and a logical theory constraining the completion. More formally, 
the input to this problem is a set $\Sigma$ of \gnfo sentences, a finite structure $F$, and 
a positive existential formula $Q$. The goal is to determine the values of $Q$ that hold in 
every structure extending the interpretations of relations in $F$ and satisfying $\Sigma$.
These values are sometimes referred to as ``the certain answers to $Q$ under $\Sigma$''. 
The complexity of open world query answering has already been identified for several \gnfo-based
languages in \cite{bbo}. Here we show that \gnfo sentences that are equivalent to a set of TGDs
have additional attractive properties from the point of view of open world query answering. 
Specifically, we extend and correct results of Baget et.~al.~\cite{bagetconf} by showing that 
the certain answers can always be determined by evaluating a sentence in a small fragment of 
(guarded negation) fixpoint logic, Guarded Negation Datalog, for which boundedness was shown 
decidable in~\cite{bbo}. From this we conclude that first-order definability of certain answers 
of \gnfo TGDs is decidable.

An extended abstract of the present paper appeared in \cite{mfcs14} and
a journal version in \cite{jsl}.
This article contains revised versions of the proofs in Section \ref{sec:rewritedep}.
Related work both prior to and subsequent to \cite{mfcs14} is discussed in Section~\ref{sec:conc}.

{\bf Organization:} Section \ref{sec:prelim} contains preliminaries.
Section \ref{sec:preserve} looks at rewriting for restricted fragments of $\gnf$,
while Section \ref{sec:interpol} looks at rewriting of queries with respect to views,
via results on Craig interpolation and Beth definability.
Section \ref{sec:rewritedep} presents our results on rewriting 
the certain answers of conjunctive queries with respect to  $\gnf$ TGDs.
Section \ref{sec:conc} covers conclusions and related work.

\section{Definitions and Preliminaries} \label{sec:prelim}
We work with fragments of first-order logic (FO) with equality and with its usual semantics,
restricting attention to finite signatures consisting of relation symbols and constant symbols 
and no function symbols.

We assume familiarity with basic notions from model theory, 
such as a \emph{reduct} of a structure (restricting the signature), 
an \emph{expansion} of a structure, and a \emph{type} (a satisfiable set 
of formulas in a collection of variables, possibly with parameters from a structure); 
and will only rely on material that can be found in the first few chapters of 
a standard model theory textbook, such as Chang and Keisler~\cite{ChangKeisler}. 
For example, we will make use of the Compactness Theorem and work with saturated 
elementary extensions. We briefly review the notion of saturation that we need in this work. 
A structure $\frakB$ is an \emph{elementary extension} of a structure $\frakA$, 
denoted $\frakA \preceq \frakB$, if $\frakB$ is an extension of $\frakA$ and every FO sentence 
with parameters from $\frakA$ that is true in $\frakA$ is also true in $\frakB$.
A structure $\frakA$ is \emph{$\omega$-saturated} if for every set of formulas $\Gamma(\vec x)$
(where $\vec x=x_1, \ldots, x_n$)
containing finitely many parameters from $\frakA$, if every finite subset of $\Gamma(\vec x)$ 
is realized by some $n$-tuple in $\frakA$, then the entire set $\Gamma(\vec x)$ is realized by 
an $n$-tuple in $\frakA$.
The conclusion means that there is a tuple $\vec c$ of elements of the domain of $\frakA$
such that $\frakA \models  \gamma(\vec c)$ for all $\gamma(\vec x) \in \Gamma(\vec x)$. 
A first-order structure is \emph{recursively saturated} if the conclusion above holds
when the collection $\Gamma$ is further required to be recursive (or, in other words,~decidable).
A basic result in model theory is that every structure has an $\omega$-saturated
elementary extension, and every countable structure (in a countable signature)
has a countable recursively-saturated elementary extension.

A
\emph{homomorphism} $h: \frakA \to \frakB$ between structures $\frakA$ and $\frakB$
is a map from the domain of $\frakA$ to the domain of $\frakB$
that preserves the 
relations
(i.e., $(a_1, \ldots, a_n)\in R^\frakA$ implies 
$(h(a_1), \ldots, h(a_n))\in R^\frakB$)
as well as the interpretation of all constant symbols 
(i.e., $h(c^\frakA) = c^\frakB$).

The primary focus of this paper is on finite structures. Finite model theory is 
concerned with logical semantics restricted to finite structures. When working with 
both classical and finite model semantics additional care must be taken to make it 
clear in each instance which semantics is meant. Crucially, both \gfo and \gnfo 
possess the finite model property (every satisfiable sentence has a finite model), 
which for most purposes voids the distinction between the two semantics and allows us 
to employ classical tools in the service of finite model theory. But at times, 
when working with different formalisms, we will need to be more specific as to 
which semantics is meant. We shall use the shorthand ``(Both classically and in the finite.)'' 
in formal assertions to signify that the statement holds equally true when 
semantic entailment is unrestricted and when it is restricted to finite structures.

\myparagraph{Database query languages and constraint languages}
One  motivation for this work is to explore how well \gnfo is suited for 
database applications. Accordingly, we will work with several logics and 
that are common in database theory,  introduced below. 
\begin{compactitem}
\item \emph{Existential FO}, comprises formulas $\exists x_1 \ldots x_n ~ \phi$, 
      where $\phi$ is quantifier-free.
\item \emph{Conjunctive queries} (CQ), are the subset of existential FO 
      where the quantifier-free kernel $\phi$ above does not contain 
      disjunction or negation.
      Equivalently, these are the first-order formulas in prenex normal form
      built up using only $\wedge$ and $\exists$.
      A \emph{boolean} conjunctive query is a CQ without free variables, that is,
      expressed as a FO sentence. 
\item \emph{Acyclic conjunctive queries} form an algorithmically well-behaved 
      subclass of conjunctive queries~\cite{Yannakakis81,FFG02,GLS03}.
      The standard definition of acyclic CQ involves the notions of hypergraph acyclicity 
      and hypergraph structure of a CQ~\cite{GLS03}. We will not need to directly use
this  definition, but only the following equivalent characterization,
which generalizes one in ~\cite{GLS03} for boolean acyclic CQs.
A formula $\phi$ is \emph{answer-guarded} if it is of the form $\phi(\vec x) = R(\vec x) \wedge \phi'$
 for some $\phi'$ and relation symbol $R$.  Then we have the following alternative characterization
of acyclic answer-guarded CQs:

\begin{fact} \label{fact:acyclicCQ}
An answer-guarded conjunctive query is acyclic iff it is equivalent to a positive existential \gfo formula.
\end{fact}

\item \emph{Tuple-generating dependencies} (TGD) are sentences of the form
      $$
      \forall \vec x \, \left( \phi(\vec x) \rightarrow \exists \vec y \, \rho(\vec x, \vec y) \right)
      $$
      where $\phi$ and $\rho$ are conjunctions of positive relational atoms (no equalities), 
      and every variable from $\vec{x}$ occurs in at least one conjunct of $\phi$. 
$\phi$ is called the \emph{body} of the TGD, while $\rho$ is referred
to as the \emph{head}.
\end{compactitem}

In addition to the above fragments of FO, some of our arguments involve \emph{Datalog}
a language that extends positive-existential FO with a fixpoint mechanism. 
Datalog programs use a signature that is partitioned into ``intensional relations'', representing the results 
of a fixpoint computation, and ``extensional relations'' that represent an input structure. 
In terms of second-order logic, intensional relations can be viewed as second-order variables, while 
extensional relations
are part of the signature of the structure over which the program is being evaluated. 
A Datalog program $\Pi$ consists of rules $R(x_1\ldots x_n) ~\datalogarrow~  \phi$, 
where $R$ is an intensional relation and $\phi$ is a CQ over intensional and extensional relations,
such that each variable $x_i$ occurs in at least one conjunct of $\phi$. 
Associated to the program $\Pi$ is an operator that takes as input a structure $\frakA$ 
in the extended signature that
includes both the extensional and intensional relations and returns a structure $\frakA'$ over 
the same extended signature. 
$\frak{A}'$ agrees with $\frak{A}$ on all extensional relations.
For each intensional relation $R$, $R_{\frak{A}'}$ is the 
set of $n$-tuples obtained by evaluating a rule of $\Pi$ of the form 
$R(x_1 \ldots x_n) ~ \datalogarrow ~  \phi$ (that is, evaluating $\phi$ in $\frakA$ 
and projecting on variables $x_1\ldots x_n$). This ``immediate consequence'' operator on structures is monotone, and thus has a unique least fixpoint.
The result of evaluating a program $\Pi$ on a structure $\frakA$ is the least fixpoint (starting with 
all intensional relations empty). 
Given a distinguished intensional predicate $P$ (the \emph{goal} predicate), 
the \emph{output} of a Datalog program is the set of tuples belonging to the goal predicate in the least fixpoint.
Datalog can be viewed as the positive-existential fragment of least-fixpoint logic.

Abiteboul, Hull, and Vianu~\cite{AHV} is a good reference for all of these languages. 

One subtle but notable difference in the treatment of query languages in the 
database literature and the logic literature concerns the relationship between
database instances and (finite) first-order structures. 
A \emph{database instance} (or simply \emph{instance}) $I$ for a signature $\tau$, 
assigns to every relation symbol $R\in\tau$ of arity $n$ 
a collection of $n$-tuples, and to every constant symbol $c$ a value, 
called the \emph{interpretation} of $R$, and respectively of $c$, in $I$. 
A \emph{fact} over a signature $\tau$ is an  expression $R(a_1 \ldots a_n)$, 
where $R$ is a relation symbol and $a_1 \ldots a_n$ are values.
An interpretation of a relation $R$ can be equivalently considered as a set of facts, 
namely the facts of the form $R(a_1 \ldots a_n)$ where $(a_1, \ldots, a_n)$ belongs to the interpretation
of $R$. The \emph{active domain} of an instance or a structure
is the set of values that participate in some fact,
or, in other words, the union of the one-dimensional projections of the relations.
We write $\adom(\frakA)$ for the active domain of $\frakA$.
Note the difference between an instance and a relational structure: a relational structure
is defined over an explicitly given domain, which can contain any number of ``inactive'' 
elements. Two structures can thus correspond to the same instance while having different domains. 
In database theory one is typically interested in \emph{domain-independent} formulas, that is, 
formulas that do not distinguish between structures corresponding to the same instance. 
For example the sentence $\exists x ~ U(x)$ is domain-independent, while $\forall x ~ U(x)$ is not.
Both CQs and Datalog are languages defining only domain-independent formulas.
In parts of this work, we will deal with logical formulas that are domain-independent.
For a domain-independent sentence $\phi$ we can talk about $\phi$ ``being true on instance $I$'',
and similarly give semantics to domain-independent formulas in terms of instances rather than structures.  Thus if we are dealing with questions about domain-independent
formulas, it will often be convenient to perform constructions that
form instances from instances, rather than constructions that  form structures
from structures.
A homomorphism $h: I\to J$ between instances $I$ and $J$ is defined
as with structures, but $h$ is now defined on the active domain of $I$, and is required
to preserve the interpretation of the relations as well as any constants occurring
in the active domain of $I$.

Given two structures $\frakA, \frakB$ over the same signature $\tau$, we write
$\frakA\subinst \frakB$
if the two structures agree on the interpretation of the constant symbols, and,
for every relation $R\in\tau$,  $R^\frakA \subseteq R^\frakB$. This
can be thought of as a weak version of the usual substructure relation,
where we do not require the substructure to be induced by taking a subset of the domain.
Since the definition does not refer to the domains of the  structures
$\frakA, \frakB$, it is clearly also applicable to instances.

To every CQ $q(\textbf{x}) = \exists \textbf{y} \bigwedge_i \alpha_i$ of signature $\tau$ 
one can associate the $\tau$-instance $\canoninst(q)$, the \emph{canonical instance} associated to $q$:
 the  active domain of $\canoninst(q)$ consists
of the set of variables and constants
 occurring in $q$ and the  facts are the literals 
$\alpha_i$. 
Evaluation of a CQ can be restated in terms of homomorphisms from $\canoninst(Q)$:
for every $n$-ary CQ $q(x_1\ldots,x_n)$ and every $n$-tuple $\textbf{a}$ of an instance
$I$ we have that $I \models q(\textbf{a})$ iff there exists a homomorphism 
 $h:(\canoninst(q),\textbf{x}) \to (I,\textbf{a})$~\cite{CM77}. 

\myparagraph{The Guarded-Negation Fragment}
The Guarded Negation Fragment ($\gnf$) is a syntactic fragment of first-order logic, from which it inherits the usual semantics. 
The formulas of \gnf are built up inductively according to the grammar\footnote{In practice, 
the parentheses are often omitted and parsing ambiguity is resolved with the help of the 
standard order of precedence of logical connectives. 
}
\begin{align*}
\phi ::= R(t_1, \ldots, t_n) | ~ t_1=t_2 ~ | ~ \exists x\,(\phi) ~ | 
~ (\phi \vee \phi )~ | ~ (\phi \wedge \phi) ~ | ~ (\alpha \wedge \neg \phi)
\end{align*}
where $R$ is a relation symbol, each $t_i$ is a variable or a constant symbol,
and, in the last clause, $\alpha$ is an atomic formula (possibly an equality)
in which all free variables of the negated formula $\phi$ occur.  
That is, each use of negation must occur conjoined with an atomic formula 
that contains all the free variables of the negated formula.
The atomic formula $\alpha$ that witnesses this is called a \emph{guard} for $\neg\phi$.
Since we allow equalities as guards,  every  formula 
  with at most one free variable can be trivially guarded, and we
  often write $\neg\phi$ instead of $\left((x=x)\land\neg\phi\right)$, when
  $\phi$ has no free variables besides (possibly) $x$.
For $\tau$ a signature consisting of constant symbols and relation symbols, 
$\gnf[\tau]$ denotes the $\gnf$ formulas in signature $\tau$.

\gnfo should be compared to the \emph{Guarded Fragment} (\gfo) of first-order 
logic~\cite{AvBN98JPL,Gr99JSL} typically defined via the grammar
\begin{align*}
\phi ::= R(t_1, \ldots, t_n) ~ | ~ t_1 = t_2 ~ | ~ \exists \vec{x} \left(\alpha \wedge \phi\right)~ |
~ (\phi \vee \phi) ~ | ~ (\phi \wedge \phi) ~ | ~ \neg \phi
\end{align*}
where, in the third clause, $\alpha$ is again an atomic formula in which 
all free variables of $\phi$ occur (and $\vec{x}$ may be a sequence of variables).
Note that, in \gfo formulas, all quantification must occur in conjunction
with a guard, while there is no restriction on the use of negation.

Since $\gnfo$ is closed under conjunction and existential
quantifications, every conjunctive query is expressible in \gnfo.
It is not much more 
difficult to verify that every \gfo sentence can also be equivalently expressed in \gnfo~\cite{BtCS15jacm}. 
Turning to fragments of first-order logic that are common in database theory, 
consider \emph{guarded tuple-generating dependencies}: that is, sentences of the form 
$$
  \forall \vec x \left( R(\vec{x}) \wedge \phi(\vec{x}) 
    \rightarrow \exists \vec{y}\, \psi(\vec{x}, \vec{y}) \right) \ .
$$
By simply writing out such a sentence using $\exists, \neg, \wedge$, 
one sees that it is convertible to a \gnfo sentence. 
In particular, every \emph{inclusion dependency} (i.e.~every formula 
$\forall \vec{x} \left( R(\vec{x}) \rightarrow \exists \vec{y}\, S(\vec{x}, \vec{y}) \right)$,
where the atomic formulas $R(\vec x)$ and $S(\vec{x}, \vec{y})$ have no constants 
and no repeated variables) is expressible in \gnfo. 
As mentioned in the introduction, many of the common dependencies used to describe 
relationships between schemas (e.g.~ see
~\cite{dataint,FKMP05}) are expressible in \gnfo.
In addition, many of the common 
description logic languages used in the semantic web (e.g.~$\mathcal{ALC}$ 
and $\mathcal{ALCHIO}$~\cite{dl}) are known to admit translations into \gfo
and hence into \gnfo.

We will frequently make use of the key result from \cite{BtCS15jacm} 
showing that \gnfo is decidable and has the finite model property:

\begin{theorem} \label{thm:gnfsat} 
A \gnfo formula is satisfiable over all structures iff it is satisfiable over finite structures. 
Satisfiability and validity of \gnfo is decidable (and $\twoexptime$-complete).
\end{theorem}

It was shown in~\cite{bbo} that \gnfo can be equivalently restated as a fragment of 
Codd's relational algebra, and of the standard database query language SQL. 
More specifically, in~\cite{bbo}, a fragment of relational algebra, called 
Guarded-Negation Relation Algebra (\gnra) is introduced, and is shown to capture
domain-independent GNFO. It is worth noting also that we can actually decide whether 
a given \gnfo formula is domain-independent (and hence whether it can be converted to \gnra). 
This is in contrast to the well-known fact that domain-independence is undecidable 
for first-order logic~\cite{AHV}. To see the decidability, we simply note that the statement 
expressing that a \gnfo formula is domain-independent can be expressed as the validity 
of a \gnfo sentence: the sentence is formed by introducing relations for the two
domains, and relativizing quantification to those domains. We can then apply Theorem~\ref{thm:gnfsat}
to this sentence.

Note that if we have two  \gnfo open formulas $\phi_1(\vec x)$ and $\phi_2(\vec x)$, 
the sentence stating that they are equivalent, or that one implies the other, 
is not necessarily a \gnfo sentence. This does hold, however, if $\phi_1$ and $\phi_2$ 
are answer-guarded.
We will need to require 
answer-guardedness in some of our results involving open formulas.%
\footnote{Note, however, that the equivalence problem and the entailment problem are 
decidable in $\twoexptime$ even  for \emph{non}-answer-guarded \gnfo formulas 
(as follows from a easy reduction
in which free variables are replaced by constant symbols). See, for example, Corollary \ref{cor:decidefo}.}
Most results about \gnfo sentences trivially generalize to answer-guarded \gnfo formulas.
For instance, the observation from \cite{BtCS15jacm} that every \gfo sentence can be 
equivalently transcribed into \gnfo extends to answer-guarded \gfo formulas.

\myparagraph{Guarded sets and tuples}
Let $\frakA$ be a structure and  $e_1,\ldots,e_k$ be the interpretation
of all constants in the signature of $\frakA$.
A subset $X$ of the domain of $\frakA$
is \emph{guarded}
if there is a fact (in some relation) in which all members of
$X\setminus\{e_1,\ldots,e_k\}$ occur together.
We will sometimes apply the same notion to tuples: a tuple of values from the domain
of a structure is guarded (in the structure), if the set of all elements of the tuple is guarded.
Note that an answer-guarded query can only be satisfied by guarded tuples.

\myparagraph{Tree-like model property} Satisfiable $\gfo$ formulas always have models
that are ``tree-like'': this is the
\emph{tree-like model property} of \gfo~\cite{AvBN98JPL,Gr99JSL}.
For any  relational structure $\frakA$ with constants, and any  guarded tuple $\vec{a}$ 
there is a
\emph{guarded unravelling} ~\cite{AvBN98JPL} $(\frakA^\ast_\vec{a}, \langle \vec{a}\rangle)$
of  $\frakA$ at $\vec{a}$, a structure and tuple such that:
\begin{compactitem}
\item[(i)] $\frakA^\ast_\vec{a}$ is \emph{tree like} in the sense that
    it has a tree decomposition with guarded bags~\cite{GO14survey};
\item[(ii)]
    $\frakA^\ast_\vec{a} \models \varphi(\langle\vec{a}\rangle)$ if and only if $\frakA\models \varphi(\vec{a})$
    for all $\varphi(\vec{x}) \in \gfo$.
\end{compactitem}

We conclude this section by recalling an important result about
approximating arbitrary answer-guarded conjunctive queries by conjunctive
queries that are in $\gfo$, which is proven using the unravellings above.

Paraphrasing~\cite{BGO14lmcs} we define the \emph{treeification} $T(q)$ 
of an answer-guarded CQ $q$ as the collection of minimal acyclic CQ that imply $q$. 
From~\cite{BGO14lmcs} we know that $T(q)$ is finite if the signature is finite.
We will thus sometimes identify  the treeification with the (answer-guarded) UCQ $\bigvee T(q)$.

The next fact is a simple consequence of the definition of treeification and of the
properties of guarded unravellings. It was first observed in~\cite{BGO14lmcs} in the 
case of  boolean CQs, but the same reasoning applies to answer-guarded CQs.

\begin{fact}[Treeification]\label{fact:treeification} 
For every answer-guarded CQ $q(\vec{x})$, every structure $\frakA$ and guarded tuple $\vec{a}$ of $M$
it holds that $\frakA^\ast_\vec{a} \models q(\langle\vec{a}\rangle)$ 
iff $\frakA^\ast_\vec{a} \models \bigvee T(q)(\langle\vec{a}\rangle)$. 
Consequently, for every answer-guarded \gfo formula $\phi(\vec{x})$ 
and answer-guarded conjunctive query $q(\vec{x})$ it holds that  
$\phi(\vec{x}) \models q(\vec{x})$ iff $\phi(\vec{x}) \models \bigvee T(q)(\vec{x})$.
\end{fact}
We note that guarded unravellings are typically infinite and that it takes 
considerably more work to show that the last claim remains valid 
when restricting attention to finite structures~\cite{BGO14lmcs}. This 
claim is what underpins the argument in~\cite{BtCS15jacm} establishing the finite model property of \gnfo.

\section{Characterization and Preservation theorems} \label{sec:preserve}

Preservation theorems are results showing that every property definable 
within a certain logic and which additionally satisfies some important 
semantic invariance can be expressed by a formula in the logic 
whose syntactic form guarantees that invariance.
One example from classical model theory is the \L{}o\'s-Tarski theorem, 
stating that a property of structures definable in first-order logic
is  definable by a universal formula if and only if  it  is
preserved under taking substructures. A second example is the Homomorphism Preservation theorem,
stating that a property of structures definable in first-order logic
is expressible by an existential positive sentences if and only if it is
preserved under homomorphism \cite{ChangKeisler}.
One can consider the ``finite model theory analogs'' of each of these statements: 
for example, the finite model theory analog of \L{}o\'s-Tarski would be 
that a property of finite structures definable in first-order logic that
is preserved under taking substructures must be definable by a universal formula 
of  first-order logic. This analog is known to fail~\cite{EF99}. Rossman~\cite{ross} 
has shown that the finite analog of the Homomorphism Preservation theorem does hold.

A well-known preservation theorem from modal logic is Van~Benthem's theorem, 
stating that basic modal logic can express precisely the properties expressible 
in first-order logic invariant under bisimulation~\cite{vdb}. Rosen~\cite{rosen} 
has shown that Van Benthem's theorem also remains valid if one restricts attention 
to finite structures, cf.~also~\cite{Otto04APAL}. Analogous results on arbitrary structures 
have been established for both \gfo~\cite{AvBN98JPL} and \gnfo~\cite{BtCS15jacm}.
In the context of finite model theory, Otto~\cite{Otto12jacm,Otto13apal} provided 
Van Benthem-style characterizations of \gfo and of the ``$k$-bounded fragment of $\gnf$'' 
indexed by a number $k$.
Central to these results are the notions of \emph{guarded bisimulation} 
and \emph{guarded negation bisimulation} that play similar roles in the model theory 
of \gfo, respectively, \gnfo as does bisimulation in the model theory of modal logic. 
For a comprehensive survey the interested reader should turn to~\cite{GO14survey}.


\subsection{Characterizing \gnfo within FO}
We first look at the question of characterizing \gnfo as a fragment of first-order 
logic invariant under certain simulation relations. In~\cite{BtCS15jacm} 
\emph{guarded-negation bisimulations} (\emph{GN-bisimulations}) were introduced, 
and it was shown that \gnfo expresses the first-order logic properties that are 
invariant under GN-bisimulations. 
A related characterization over finite structures
for the $k$-variable fragment of $\gnfo$ is given in \cite{Otto13apal}.
Here we will work over all structures, giving a characterization theorem for a 
simpler kind of simulation relation, which we call a \emph{strong GN-bisimulation}.
We will use this characterization as a basic tool throughout the paper: 
to show that a certain formula is equivalent to one in \gnfo, to argue 
that two structures must agree on all \gnfo formulas and to amalgamate 
structures that cannot be distinguished by \gnfo sentences in a sub-signature. 
The many uses of strong GN-bisimulations suggest that it is really \emph{the} 
right equivalence relation for \gnfo.

Recall that a homomorphism from a structure $\frakA$ to a structure $\frakB$ 
is a map from the domain of $\frakA$ to the domain of $\frakB$ 
that preserves the relations as well as the interpretation of the constant symbols. 
Recall also that a set, or tuple, of elements from a structure $\frakA$ is \emph{guarded} 
in $\frakA$ if there is a fact of $\frakA$ that contains all elements within the fact
except possibly those that are the interpretation of some constant symbol. 

\begin{definition}[Strong GN-bisimulations]
  A strong GN-bisimulation between structures $\frakA$ and $\frakB$
  is a non-empty collection $Z$ of pairs $(\mathbf{a},\mathbf{b})$ 
  of guarded tuples of elements of $\frakA$ and of $\frakB$, respectively,  
  such that for every $(\mathbf{a},\mathbf{b}) \in Z$: 
  \begin{compactitem}
  \item there is a homomorphism $h\colon\frakA\to\frakB$ such that 
    $h(\mathbf{a})=\mathbf{b}$ and ``$h$ is compatible with $Z$'', meaning that 
    $(\mathbf{c},h(\mathbf{c}))\in Z$ for every guarded tuple $\mathbf{c}$ in $\frakA$. 
  \item there is a homomorphism $g\colon\frakB\to\frakA$ such that 
    $g(\mathbf{b})=\mathbf{a}$ and ``$g$ is compatible with $Z$'', meaning that 
    $(g(\mathbf{d}),\mathbf{d})\in Z$ for every guarded tuple $\mathbf{d}$ in $\frakB$.
  \end{compactitem}
  We write $(\frakA,\textbf{a})\rightarrow^s_{GN} (\frakB,\textbf{b})$ 
  if the map $\textbf{a}\mapsto\textbf{b}$ extends to a homomorphism 
  from $\frakA$ to $\frakB$ that is compatible with some strong GN-bisimulation 
  between $\frakA$ and $\frakB$. Note that, here, $\textbf{a}$
  and $\textbf{b}$ are not required to be guarded tuples.
  We write $(\frakA,\textbf{a})\sim^s_{GN} (\frakB,\textbf{b})$ if, furthermore,
  $\textbf{a}$ is a guarded tuple in $\frakA$ (in which case we also have 
  that  $(\frakB,\textbf{b})\sim^s_{GN} (\frakA,\textbf{a})$).
  These notations can also be indexed by a signature $\sigma$, in which case
  they are defined in terms of $\sigma$-reducts of the respective structures.
\end{definition}

It is easy to see that if there exists a strong GN-bisimulation between two structures, 
then the respective substructures consisting of the elements designated by 
constant symbols must be isomorphic.

The key distinction between strong GN-bisimulation and the GN-bisimulation 
of~\cite{BtCS15jacm} is that the homomorphisms whose existence is postulated 
in the back-and-forth properties of GN-bisimulation are only required to be ``local'', 
that is, defined on arbitrary finite neighbourhoods of the guarded tuple in question, 
while our definition above asks for a single ``global'' homomorphism that is defined 
on the entire domain of the respective structure, i.e.~one that is uniformly appropriate 
for all neighbourhoods according to the requirements of GN-bisimulations of~\cite{BtCS11icalp}. 
This is a very significant strengthening of requirements, which makes strong GN-bisimulation 
more powerful as a tool in our proofs.  

Another distinction between the notions is that while GN-bisimulations are only defined 
on guarded tuples, our notion of strong GN-bisimulation is meaningful on arbitrary tuples. 
It is an equivalence relation on guarded tuples, but is asymmetric on general tuples.

In~\cite{BtCS15jacm} it was shown that \gnfo corresponds to the GN-bisimulation-invariant 
fragment of first-order logic.  In light of our previous remark, it  follows that \gnfo 
formulas are also invariant under strong GN-bisimulations as far as guarded tuples are concerned. 
In fact, for arbitrary tuples one can verify via structural induction on the construction of formulas 
that all \gnfo formulas are preserved by strong GN-bisimulations. That is, one can
show
that $\to^s_{GN}$ implies $\Rrightarrow_{GN}$, where the notation 
\[
  (\mathfrak{A},\textbf{a})\Rrightarrow_{GN} (\mathfrak{B},\textbf{b})
\]
expresses that, for every \gnfo formula $\phi(\textbf{x})$, 
$\mathfrak{A}\models\phi(\textbf{a})$ implies $\mathfrak{B}\models\phi(\textbf{b})$.

Strong GN-bisimulations will play a key role in our remaining results.
Informally, when we want to show that a \gnfo formula $\phi$ can be replaced by
another simpler $\phi'$, we will often justify this by showing that
an arbitrary model of $\phi$ can be replaced by a strongly bisimilar
structure where $\phi'$ holds (or vice versa).

Our first ``expressive completeness'' result characterizes \gnfo as
the fragment of first-order logic that is preserved by strong GN-bisimulations.

\begin{theorem}\label{thm:strongbisim}
A first-order formula $\phi(\textbf{x})$ is preserved by $\rightarrow^s_{GN}$ 
(over all structures) iff it is equivalent to a \gnfo formula.
\end{theorem}


\newcommand{\vequiv}{\mathbin{\protect\rotatebox[origin=c]{90}{$\equiv$}}}
\newcommand{\velext}{\mathbin{\protect\rotatebox[origin=c]{90}{$\preceq$}}}

The proof of the of Theorem~\ref{thm:strongbisim} relies 
on the following lemma. 
Further, in the remainder of the paper, we will make use of the lemma directly.
For example, the second part of the lemma
 will be instrumental in our proof of Craig Interpolation for \gnfo
presented in Section~\ref{sec:interpol}.

The first part of the lemma will be used in the
``easy direction'' of Theorem ~\ref{thm:strongbisim}: it formalizes
the notion that strong bisimulation preserves $\gnfo$ formulas.
The second part of the lemma will be used in the harder direction of
Theorem  ~\ref{thm:strongbisim}. It asserts that $\Rrightarrow_{GN}$ can always 
be lifted to $\to^s_{GN}$ by passing from a pair of structures to suitable 
elementary extensions. 
The second part
will be established using the technique of \emph{recursively saturated models}~\cite{ChangKeisler}.


\begin{lemma}~ \label{lem:lifting}
\begin{compactenum}
\item If
  $(\mathfrak{A},\textbf{a})\rightarrow^s_{GN[\sigma]}(\mathfrak{B},\textbf{b})$
  then $(\mathfrak{A},\textbf{a})\Rrightarrow_{GN[\sigma]} (\mathfrak{B},\textbf{b})$. 
\item If $(\mathfrak{A},\textbf{a})\Rrightarrow_{GN[\sigma]} (\mathfrak{B},\textbf{b})$ 
  and both structures are countable, 
  then there are countable elementary extensions  
  $(\widehat{\mathfrak{A}},\textbf{a})$ and $(\widehat{\mathfrak{B}},\textbf{b})$, 
  respectively, such that
  $(\widehat{\mathfrak{A}},\textbf{a})\rightarrow^s_{GN[\sigma]}(\widehat{\mathfrak{B}},\textbf{b})$.
\end{compactenum}
\end{lemma}

\begin{proof}
  The first part can be proved by a straightforward formula
  induction. 
  For the second part, we will use countable recursively saturated
  structures. 

  Consider the pair of countable structures $(\mathfrak{A},\mathfrak{B})$ 
  viewed as a single structure over an extended signature with additional
  unary predicates $P$ and $Q$ to denote the domain of
  $\mathfrak{A}$ and of $\mathfrak{B}$, respectively. 
  Let $(\widehat{\mathfrak{A}},\widehat{\mathfrak{B}})$ be any 
  countable recursively saturated elementary extension of
  $(\mathfrak{A},\mathfrak{B})$.  Let $Z$ be the collection of all
  pairs of guarded tuples of $\widehat{\frakA}$ and $\widehat{\frakB}$ 
  that are $\gnf$-indistinguishable. 
  To establish the lemma, we need to show that $Z$ is a strong
  GN-bisimulation, and that the partial map
  $\textbf{a}\mapsto\textbf{b}$ extends to a homomorphism that is
  compatible with $Z$. Both follow directly from the following claim.

\begin{trivlist}
\item \textbf{Claim.} Every finite partial map $f$ from $\widehat{\mathfrak{A}}$ 
  to $\widehat{\mathfrak{B}}$, or vice versa, that preserves truth of all 
  \gnfo-formulas, can be extended to a homomorphism $f'$ compatible with~$Z$. 
\item \textbf{Proof of claim.}
  We assume that $f$ is a finite partial map from $\widehat{\mathfrak{A}}$ to
  $\widehat{\mathfrak{B}}$; the other direction is symmetric. 
  Fix an enumeration $c_1, c_2, \ldots$ of the  (countably many)
  elements of the domain of $\widehat{\mathfrak{A}}$ that are not in
  the domain of $f$. We will define a sequence of finite partial maps
  $f=f_0\subseteq f_1 \subseteq f_2\subseteq \cdots$ such that
  $dom(f_{i+1})=dom(f_i)\cup \{c_{i+1}\}$, and such that each $f_i$
  preserves truth of all $\gnf$ formulas. It then follows that
  $\bigcup_i f_i$ is a homomorphism extending $f$ and compatible with $Z$.

  It remains only  to show how to construct $f_{i+1}$ from $f_i$. Here, we
  use the fact that  $(\widehat{\mathfrak{A}},\widehat{\mathfrak{B}})$
  is recursively saturated. Let $\textbf{c}$ be an enumeration of the
  domain of $f_i$, and $\textbf{d}$ an enumeration of the range of $f_i$, 
  corresponding to the enumeration of  $\textbf{c}$,
  and let $\Sigma(x)$ be the set of all first-order formulas of the form
  \[ 
    \phi(\textbf{c}, c_{i+1})\to\phi(\textbf{d}, x)
  \]   
  where $\phi(\textbf{c}, c_{i+1})$ is a $\gnf$ formula with parameters 
  $\textbf{c}$ and $c_{i+1}$, and $\phi(\textbf{d}, x)$ is obtained by 
  replacing each parameter in $\textbf{c}$ by its $f_i$-image, and 
  replacing $c_{i+1}$ by $x$. In the above definition of $\Sigma(x)$ we
  only consider formulas $\phi(\textbf{c}, c_{i+1})$  that belong to $\gnf$
  even when the parameters $\textbf{c}, c_{i+1}$ are treated as free variables 
  (thereby excluding formulas such as $c_1\neq c_2$).

  The set $\Sigma(x)\cup\{Q(x)\}$ is clearly a recursive set.  
  From the fact that $f_i$ preserves truth of \gnfo formulas 
  it follows that every finite subset of $\Sigma(x)\cup\{P(x)\}$ 
  is realized in $(\widehat{\mathfrak{A}},\widehat{\mathfrak{B}})$. 
  Note that in the argument above we are only relying on the 
closure of \gnfo 
  under conjunction and existential quantification.

  By compactness, therefore, $\Sigma(x)\cup\{Q(x)\}$ is consistent and, 
  by virtue of recursive saturation, it is realized by some element $d_{i+1}$. 
  It follows from the construction that the partial map 
    $f_{i+1} = f_i \cup \{(c_{i+1},d_{i+1})\}$
  preserves truth of all \gnfo formulas. \hfill $\dashv$
\end{trivlist}

This concludes the proof of the lemma.\end{proof}

\myparagraph{Proof of Theorem~\ref{thm:strongbisim}}
We prove only the harder direction, following the template often used 
in preservation theorems in classical model theory. 
Let $\phi(\textbf{x})$ be preserved by $\rightarrow^s_{GN}$,
and let $\Psi(\textbf{x})$ be the set of all \gnfo formulas it entails.
Thanks to compactness, it is enough to show that $\Psi(\textbf{x})\models\phi(\textbf{x})$.

Let $\mathfrak{B}\models\Psi(\textbf{b})$, and let
$\Gamma_{\mathfrak{B},\mathfrak{b}}(\textbf{x})$ be the set of all
negations of \gnfo formulas false of $\mathbf{b}$ in $\mathfrak{B}$.
We claim that $\Gamma_{\mathfrak{B},\mathbf{b}}(\textbf{x})\cup\{\phi(\mathbf{x})\}$
is consistent. Suppose it were not consistent. Then by the
Compactness Theorem we would have that $\phi(\textbf{x})$ implies
$\gamma(\textbf{x})$, where $\gamma(\textbf{x})$ is the 
negation of some finite conjunction of formulas from
$\Gamma_{\mathfrak{B},\mathbf{b}}(\textbf{x})$. It follows from the
construction of $\Gamma_{\mathfrak{B},\mathbf{b}}(\textbf{x})$ that 
$\gamma(\textbf{x})$ is (up to logical equivalence) a \gnfo formula, 
which therefore must belong to $\Psi(\textbf{x})$. This yields a 
contradiction because we have that $\mathfrak{B}\models\Psi(\textbf{b})$
and $\mathfrak{B}\not\models\gamma(\textbf{b})$.

Thus there is $\mathfrak{A}$ and $\textbf{a}$ such that
$\mathfrak{A}\models\Gamma_{\mathfrak{B},\mathbf{b}}(\textbf{a})\land\phi(\mathbf{a})$.
By construction, every \gnfo formula true of $\mathbf{a}$ in $\mathfrak{A}$
is also true of $\mathbf{b}$ in $\mathfrak{B}$. 
Note that we may assume that both $\mathfrak{A}$ and $\mathfrak{B}$ are countable.
Using Lemma~\ref{lem:lifting}, we can find elementary equivalent extensions completing 
the following diagram.
\[\begin{array}{ccc}
 (\widehat{\mathfrak{A}},\textbf{a}) & \to^s_{GN} & (\widehat{\mathfrak{B}},\textbf{b}) \\
 \velext & & \velext \\
 (\mathfrak{A},\textbf{a}) & \Rrightarrow_{GN} & (\mathfrak{B},\textbf{b})
\end{array}\]
By virtue of $\phi$ being invariant under elementary equivalence and being preserved 
by strong GN-bisimulations, we can chase it around the diagram starting from 
$\mathfrak{A}\models\phi(\mathbf{a})$ and concluding $\mathfrak{B}\models\phi(\mathbf{b})$.
Given that $\mathfrak{B}\models\Psi(\mathbf{b})$ was arbitrary, this shows 
that $\Psi(\textbf{x})\models\phi(\textbf{x})$ and so the theorem follows.
\qed

Note that our proof makes use of infinite structures
in a fundamental way. We do not claim the analogous
result for preservation over finite structures.

\medskip

We now look at characterizing the intersection of \gnfo 
with smaller fragments of  first-order logic.
We will start with tuple-generating dependencies. 

\subsection{Tuple-generating dependencies within \gnfo} \label{subsec:tgdchar}
Recall that a \emph{tuple-generating dependency (TGD)} is a sentence of the form:
\[
\forall \vec x \, \left( \phi(\vec x) \rightarrow \exists \vec y \, \rho(\vec x, \vec y) \right)
\]
where $\phi$ and $\rho$ are conjunctions of relational atomic formulas (not equalities).
TGDs arise in databases, as a way of specifying natural
restrictions on data and as a way of capturing relationships
between different datasources. They also arise
in ontological reasoning. Static analysis and query answering problems 
have motivated research to identify expressive yet computationally well-behaved classes of TGDs. 
A guarded TGD ($\gtgd$) is one in which $\phi$ includes an atomic formula 
containing all the variables $\vec x$ occurring in $\phi$. Guarded TGDs constitute an important 
class of TGDs at the heart of the Datalog$^\pm$ framework~\cite{CGL09pods,BGO14lmcs} 
for which many computational problems are decidable. 
More recently, Baget, Lecl{\`e}re, and Mugnier~\cite{baget2010} introduced 
\emph{frontier-guarded TGDs} ({\fgtgd}s), defined like guarded TGDs, but where 
only the variables occurring both in $\phi$ and in $\rho$ (the \emph{exported} variables) 
must be guarded by an atomic formula in $\phi$. 
Every \fgtgd is equivalent to a \gnfo sentence, obtained just by writing it out 
using existential quantification, negation, and conjunction.
Theorem~\ref{thm:fgtgd-char} below shows that these are \emph{exactly} the TGDs
that \gnfo can express. \looseness=-1 

We need two lemmas: one about \gnfo and one about TGDs.
For two structures $\frakA \subinst \frakB$, let us denote by $\frakB \ominus \frakA$ 
the structure 
obtained from  $\frakB$  by removing all facts 
containing only values from the active domain of $\frakA$.
We say that $\frakB$ is a \emph{squid-extension} of $\frakA$ if
\begin{enumerate} 
\item[(i)] every set of elements from the active domain of $\frakA$ 
      that is guarded in $\frakB$ is already guarded in $\frakA$; and 
\item[(ii)] $\frakB \ominus \frakA$ is a union of structures $\frakB'_i$ such that:
for two distinct $\frakB'_i$ and $\frakB'_j$ 
their
active domains overlap only in $\adom(\frakA) \cup C$, and
      each $(\adom(\frakB'_i)\cap \adom(\frakA))\setminus C$ is guarded in $\frakA$, where $C$ is 
      the set of elements of $\frakA$ named by a constant symbol. 
\end{enumerate}

Intuitively, we can think of $\frakB$ as a squid, where each $\frakB'_i$ is one of its tentacles.
We refer to the $\frakB_i$ as the tentacles, and the partition into $\frakB_i$ as
a \emph{squid decomposition} of $\frakB$.

We extend the notation to instances in the obvious way (since it does not depend
on the domain of $\frakA$ or $\frakB$).
The following lemma allows one to turn an arbitrary extension of a structure $\frakA$ 
into a squid-extension of $\frakA$, modulo strong GN-bisimulation.

\begin{lemma}\label{lem:squid}
  For every pair of structures $\frakA, \frakB$ with $\frakA\subinst \frakB$, there is a
  squid-extension $\frakB'$ of $\frakA$ and a homomorphism $h:\frakB'\to \frakB$ whose
   restriction to $\frakA$ is the identity function,
   such that $\frakB'\sim_{GN}^s \frakB$ via a
  strong GN-bisimulation that is compatible with $h$. Moreover, 
  we can choose $\frakB'$ to be finite if $\frakB$ is.
\end{lemma}

We will make use of Lemma~\ref{lem:squid} as a tool for bringing
certain conjunctive queries into a restricted syntactic form, by
exploiting the fact that, whenever a tuple from $\adom(\frakA)$ satisfies a
conjunctive query in a squid-extension $\frakB$ of $\frakA$, then we can
partition the atomic formulas of the query into independent subsets that are
mapped into different tentacles of $\frakB$.

\begin{proof}
  For every set $X$ of elements that is guarded in $\frakA$, we create a
  structure $\frakB_X$ that is a fresh isomorphic copy of $\frakB$ in which only
  the elements of $X\cup C$ are kept constant (i.e., mapped to themselves by
  the isomorphism), where $C$ is the set of all elements named by a
  constant symbol. We define $\frakB'$ to be the union of all such $\frakB_X$.
  Clearly,  $\frakB'$ is a squid-extension of $\frakA$, and the natural projection 
  $h:\frakB'\to \frakB$ is a homomorphism. Furthermore, we claim that $\frakB'\sim_{GN}^s \frakB$ 
  via a strong GN-bisimulation that is compatible with $h$. 
  The claimed strong GN-bisimulation consists of all pairs
  $(\vec{a},h(\vec{a}))$ where $\vec{a}$ is a guarded tuple of $\frakB'$.
\end{proof}

The following lemma expresses a general property of TGDs that follows from the fact that
TGDs are preserved under taking direct products of structures \cite{Fagin82}.

\begin{lemma} \label{lem:product}
(Both classically and in the finite.)
Let $\Sigma$ be any set of TGDs and suppose that 
$\Sigma\models\forall\vec{x} ~ (\phi(\vec{x})\to\bigvee_{i=1\ldots n} \exists\vec{y}_i\psi_i(\vec{x},\vec{y}_i))$,
where $\phi$ and the $\psi_i$ are conjunctions of atomic formulas.
Then $\Sigma\models\forall\vec{x} ~ (\phi(\vec{x})\to\exists\vec{y}_i\psi_i(\vec{x},\vec{y}_i))$
for some $i\leq n$.
\end{lemma}

\begin{proof}
 To simplify the presentation, we consider the case where
  $n=2$.  Let
 $$\Sigma\models\forall\vec{x}(\phi(\vec{x})\to
 \exists\vec{y}_1\psi_1(\vec{x},\vec{y}_1)\lor
 \exists\vec{y}_2\psi_2(\vec{x},\vec{y}_2))$$ and suppose for the sake
 of a contradiction that there are structures $I_1\models\Sigma$ and
 $I_2\models\Sigma$ such that $I_i\models
 \phi(\vec{a}_i)\land\neg\exists\vec{y}_i\psi_i(\vec{a}_i,\vec{y}_i)$.  Let
 $J$ be the direct product $I_1 \times I_2$, that is, the structure
 whose domain is the cartesian product of the domains of $I_1$ and
 $I_2$ and such that a tuple of pairs belong to a relation in $J$ if
 and only if the tuple of first-projections belongs to the
 corresponding relation in $I_1$ and the tuple of second-projections
 belongs to the corresponding relation in $I_2$. If a constant
  symbol denotes $a$ in $I_1$ and $b$ in $I_2$, it denotes the pair
  $(a,b)$ in $J$. Since TGDs are closed
 under taking direct products, we have that $J\models\Sigma$.
 It also follows from the construction that 
\begin{enumerate}
\item[(i)] the natural projections $h_1:J\to I_1$ and
 $h_2:J\to I_2$ are homomorphisms, and 
\item[(ii)] whenever
  $\phi(\vec{x})$ is satisfied by tuples $\vec{a}_1$ in $I_1$ and
  $\vec{a}_2$ in $I_2$, then the tuple of pairs $\vec{a}$ whose
  first-projections are $\vec{a}_1$ and whose second projections are
  $\vec{a}_2$ also satisfies $\phi(\vec{x})$ in $J$.
\end{enumerate}
  Putting this together, we obtain that
  $J\models \phi(\vec{a})\land\bigwedge_i\neg\exists\vec{y}_i\psi_i(\vec{a},\vec{y}_i)$,
  which contradicts the fact that $J\models\Sigma$.

Because $J$ is finite if both $I_1$ and $I_2$ are, the above argument 
is equally valid over finite structures as over arbitrary structures.
\end{proof}

We now return to describing our characterization of TGDs that are 
equivalent to some \gnfo sentence. Consider a TGD 
  $\rho = \forall \vec{x} \left( \beta(\vec{x}) \limp \exists \vec{z} \gamma(\vec{x}\vec{z}) \right)$.  
A \emph{specialization} of $\rho$ is a TGD of the form 
  $\rho^\theta = \forall \vec{x} \left( \beta(\vec{x}) \limp \exists \vec{z}' \gamma'(\vec{x}\vec{z}') \right)$ 
obtained from $\rho$ by applying some substitution $\theta$ 
mapping the variables $\vec{z}$ to constant symbols or to variables 
among $\vec{x}$ and $\vec{z}$.  Clearly, a specialization of 
a TGD $\rho$ entails $\rho$.
The following lemma states that as far as strong GN-bisimulation invariant TGDs 
are concerned, we can replace any TGD by specializations of it that are equivalent 
to frontier-guarded TGDs. Its proof relies heavily on the two lemmas above.

\begin{lemma} \label{lem:special} [TGD specializations]
(Both classically and in the finite.)
Let $\Sigma$ be a set of TGDs that is strong GN-bisimulation invariant 
and let $\rho$ be a TGD such that $\Sigma \models \rho$. 
Then there exists a specialization $\rho'$ of $\rho$ 
such that $\Sigma \models \rho'$, and such that $\rho'$ is logically
equivalent to a conjunction of frontier-guarded TGDs.
\end{lemma}
\begin{proof}
First we introduce the notion of a \emph{quasi-frontier guarded TGD}.
By the \emph{graph of a TGD} 
  $\rho = \forall \vec{x} \left( \beta(\vec{x}) \limp \exists \vec{z} \gamma(\vec{x},\vec{z}) \right)$ 
we mean the undirected graph whose nodes are the conjuncts of $\gamma$ 
and where two conjuncts are connected by an edge if they share an 
existentially quantified variable. 
Observe that if the graph of $\rho$ is not connected, then $\rho$ 
can be decomposed into several TGDs, one for each connected component. 
We say that $\rho$ is \emph{quasi-frontier guarded} if, for each
connected component of its graph, the set of universally quantified
variables occurring in atomic formulas belonging to that component 
is guarded by some atomic formula in the TGD body $\beta$. 
This is equivalent to saying that the decomposition into TGDs just mentioned
yields a set of frontier-guarded TGDs. 
%

We will show that, if $\Sigma$ is a set of TGDs that is strong GN-bisimulation 
invariant and $\rho$ is a TGD such that $\Sigma \models \rho$, then 
there exists a specialisation $\rho'$ of $\rho$ such that $\Sigma \models \rho'$, 
and such that $\rho'$ is quasi-frontier guarded.

Thus fix $\rho =  \forall \vec{x} \left( \beta(\vec{x}) \limp \exists \vec{z} \gamma(\vec{x},\vec{z}) \right)$
such that $\Sigma \models \rho$.


Consider any structure $J \models \Sigma$ and homomorphism $h:\canoninst(\beta(\vec{x})) \to J$. 
Let $B$ be the image of $h$. By Lemma~\ref{lem:squid}, $B$ has a squid-extension $J'$ 
such that $J'\sim_{GN}^s J$ via some strong GN-bisimulation that is compatible 
with a homomorphism $g:J'\to J$ whose restriction to $B$ is the identity function. 
Since $\Sigma$ is invariant for strong GN-bisimulations, $J'\models\Sigma$.
Therefore since $\Sigma \models \rho$,  $J'\models\rho$. In particular, $h$ can be extended to a homomorphism 
$h'$ from $\canoninst(\exists \vec{z}\gamma(\vec{x},\vec{z}))$ to $J'$.
We can extract from $h'$ a substitution $\theta$, namely the one that sends 
a variable $z_i$ to a constant symbol $c$ if $h'(z_i)$ is the interpretation of $c$ 
(if $h'(z_i)$ is the interpretation of several constant symbols we choose one arbitrarily), 
or else $\theta$ sends $z_i$ to an arbitrary $x_j$ for which $h'(z_i)=h(x_j)$ 
if there is such $x_j$, otherwise $\theta$ sends $z_i$ to $z_i$. 
Applying $\theta$ to the conjunctive query $\exists\vec{z}\gamma(\vec{x},\vec{z})$
yields another conjunctive query $\exists\vec{z}'\gamma'(\vec{x},\vec{z}')$ 
(where $\vec{z}'$ is a subset of $\vec{z})$. By construction we have that 
\[ 
  \rho' = \forall\vec{x} \left(\beta(\vec{x}) \to \exists\vec{z}'\gamma'(\vec{x},\vec{z}')\right) 
\] 
is a specialization of $\rho$ such that the CQ $\exists\vec{z}'\gamma'(\vec{x},\vec{z}')$ 
is satisfied in $J'$, hence also in $J$, under the assignment $h$ for the universally 
quantified variables $\vec{x}$. 
We first show that each $\rho'$ is quasi-frontier-guarded. 
Consider the decomposition of $\rho'$ 
$$ 
  \rho' \, \equiv \, 
  \bigwedge_j \forall\vec{x} \left(\beta(\vec{x}) \to \exists\vec{z}'\gamma'_j(\vec{x},\vec{z}')\right) 
$$
such that the graphs of 
$
  \rho'_j = \forall\vec{x} \left(\beta(\vec{x}) \to \exists\vec{z}'\gamma'_j(\vec{x},\vec{z}')\right)
$ 
enumerate the connected components of the graph of $\rho'$ and let $j$ be arbitrary. 

Note that, by construction, all existential variables $\textbf{z}'$ 
are mapped by $h'$ to elements that neither belong to $\adom(B)$ nor interpret any constant symbol: 
if $h'$ had mapped an existential variable to $\adom(B)$, then this variable would have been removed and replaced
by a universal variable.
Next note that the active domains of the tentacles of $J'$ overlap only on elements of
$\adom(B)$. Using connectivity of $\gamma'_j$, we see that
the existential variables must map to the active domain of a  single tentacle.
From connectedness of the graph of $\gamma'_j$, we know there are two possibilities:
if there are no existential variables in $\gamma'_j$, then $\gamma'_j$ consists of a single
atom. In this case the universal variables map into a guarded set of $B$. 
If there is any existential variable in $\gamma'_j$,
then  every universal variable  lies in some atom with an
existential variable. 
Since the existential variables do not map into $\adom(B)$,
it follows that the image of 
$\canoninst(\gamma'_j)$ under $h'$ 
must be  entirely contained in a single tentacle of $J'$.
Now the subset of the universally-quantified 
variables $\vec{x}$ 
occurring in $\gamma'_j$ is mapped into $B$,  since $h$ mapped into $B$ and
$h'$ extended $h$. Thus the variables $\vec{x}$  must be mapped by $h'$
to the intersection of a tentacle and the active domain of $B$, hence
(by the properties of a squid decomposition) 
again we can conclude that $\vec{x}$ maps to a guarded set of elements of $B$.  And since $h'$ agrees with $h$ on these
variables, the same statement holds with $h$ substituted for $h'$.
Since $B$ was defined as the $h$-image of $\canoninst(\beta)$,
we can conclude that the universally-quantified variables occurring in $\gamma'_j$ 
are guarded in $\beta$; that is, $\rho'_j$ is frontier-guarded. Since $j$ was arbitrary, 
this shows that $\rho'$ is indeed quasi-frontier-guarded.

Now we need to show that one such $\rho'$ is entailed by $\Sigma$.
What we have shown thus far is that  any $J$ that is satisfied by $\Sigma$ satisfies
one such $\rho'$. But there are only finitely many such $\rho'$, and thus
by Lemma~\ref{lem:product} we can conclude that $\Sigma$ entails one such $\rho'$.
\end{proof}

Suppose we apply the lemma above to each TGD is $\Sigma$.
We get a finite set of frontier-guarded TGDs whose conjunction  implies
each TGD in $\Sigma$. Further, each TGD in the set is implied
by $\Sigma$. Thus we have
obtained our first main characterization:

\begin{theorem}\label{thm:fgtgd-char}
Every \gnfo sentence that is equivalent to the conjunction
of a finite set of TGDs on
finite structures is equivalent to the same
conjunction of a finite set of TGDs on arbitrary
structures, and such a formula is equivalent (over all structures) 
to a finite set of FGTGDs.
\end{theorem}

In light of the above result, it may seem tempting to suppose that, similarly, 
guarded TGDs can express all that can be expressed both by TGDs and in \gfo. 
This is, however, not the case: the TGD 
  $\forall xyz ~\left( R(x,y) \land R(y,z) \limp P(x) \right)$ 
can be equivalently expressed in \gfo, but not by means of a guarded TGD; 
and the guarded TGD 
  $\forall x ~\left( P(x) \limp \exists yz \ E(x,y) \land E(y,z) \land E(z,x) \right)$ 
is not expressible in \gfo. Instead, we show that every property expressible 
both in \gfo and by a finite set of TGDs is in fact expressible by a finite set 
of \emph{acyclic frontier-guarded TGDs}.

Recall from Section~\ref{sec:prelim} that a CQ is answer-guarded 
if its free variables co-occur in one of its atomic sub-formulas and that 
such a CQ is acyclic if it is equivalent to a positive-existential \gfo formula.
We say that a frontier-guarded TGD 
 $\rho = \forall \vec{x}\vec{y}  
         \left( \beta(\vec{x},\vec{y}) \limp \exists \vec{z}\, \gamma(\vec{x},\vec{z}) \right)$ 
is acyclic if the answer-guarded CQ $\exists \vec{y}\beta(\vec{x},\vec{y})$ 
and the answer-guarded CQ $\exists \vec{y}\vec{z}\, \beta(\vec{x},\vec{y}) \land\gamma(\vec{x},\vec{z})$ are both acyclic. 
Note that both CQs are indeed answer-guarded, by virtue of $\rho$ being frontier-guarded.
\begin{theorem}\label{thm:gtgd-char}
Every \gfo sentence that is equivalent to a finite set of TGDs over
finite structures is equivalent 
(over all structures) to a finite set of acyclic FGTGDs.
\end{theorem}


\begin{proof}
  Let $\phi$ be any \gf sentence that is equivalent to a finite set of
  TGDs over finite structures. Then, by Theorem~\ref{thm:fgtgd-char}, 
  $\phi$ is equivalent to a finite set $\Sigma$ of FGTGDs over arbitrary structures.

Recall the notion of guarded unravelling $\frakA^*$ of a structure $\frakA$
and the notion of treeification of an answer-guarded CQ from Section \ref{sec:prelim}.
Note
that for each TGD in $\Sigma$, its  left-hand side is answer-guarded
by definition, and its right-hand side can be assumed answer-guarded as well.
Consider the set $\Sigma'$ of disjunctive GTGDs obtained by
  replacing the head and body of each TGD by
  its treeification, and expanding out the disjunction in the left-hand side.  

We claim that $\Sigma$ is equivalent to $\Sigma'$.
Note that since $\phi$ is in $\gf$, for any structure $\frakA$,
$\frakA \models \phi \leftrightarrow  \frakA^* \models \phi$.
Similarly, since $\Sigma'$ is in $\gf$,
$\frakA \models \phi \leftrightarrow  \frakA^* \models \phi$.
Thus it is enough to show equivalence of $\phi$ and $\Sigma'$ on
guarded unravellings.
But from Fact~\ref{fact:treeification} we see that each formula
is equivalent to its treeification on guarded unravellings, and so our claim is proven.

Now by Lemma~\ref{lem:product}, we obtain that each disjunctive TGD in $\Sigma'$
is  equivalent to one of the  GTGDs obtained by replacing the disjunction in
its head by   one of the  disjuncts. Since the head and body of each such
TGD are acyclic, each such TGD is acyclic.
\end{proof}

\subsection{Existential and Positive-Existential Formulas}
We turn to characterizing the existential formulas within \gnfo, 
establishing an analog of the \L{}o\'s-Tarski theorem. 

\begin{theorem} \label{thm:lostarski}
Every \gnfo formula that is preserved under extensions over finite structures 
has the same property over all structures, and such a formula is equivalent 
(over all structures) to an existential formula in \gnfo. 
Furthermore, we can decide whether a formula has this property, 
and also find an equivalent existential \gnfo formula effectively.
\end{theorem}

\begin{proof}
Let $\phi$ be a \gnf formula containing constants $\textbf{c}$ and
with free variables $\textbf{x}$. Let $\textbf{d}$ be fresh constants,
one for each variable in $\textbf{x}$. Then $\phi$ is preserved under 
extensions over finite structures iff the \gnf sentence 
 $\Phi = \bigwedge_{c\in\textbf{c}\cup\textbf{d}} P(c)\land\phi^P(\textbf{d})\to \phi(\textbf{d})$ 
is a validity over finite structures, where $\phi^P$ is the relativization of $\phi$ 
to a new unary predicate $P$. Since $\Phi$ is a $\gnf$ formula, 
it is a validity over finite structures iff it is a validity over all structures.
Also, the decidability of $\gnf$ allows us to decide this validity.

As to the effective content of the claim, note that 
once an equivalent existential formula is known to exist in \gnfo, 
we can find it by exhaustive search relying on the decidability 
of equivalence of \gnfo formulas.

By the classical  \L{}o\'s-Tarski~theorem, if a first-order formula is preserved 
under extensions over all structures, it is equivalent to an existential formula $\phi'$. 
Thus, to complete the proof, it suffices to show that every \gnfo formula $\phi$ 
that is equivalent to an existential formula $\phi'$ is also equivalent to an 
existential \gnfo formula $\phi''$. We can assume that $\phi$ is satisfiable,
since otherwise it is clearly equivalent to a \gnfo formula.
We can convert $\phi'$ into the form $\bigvee_i \phi'_i$, 
where $\phi'_i(\textbf{x}) = \exists \textbf{y} \left( \varepsilon'_i \wedge \bigwedge_j \psi'_{ij}\right)$ 
with each $\psi'_{ij}$ a possibly negated relational atom 
and where each $\varepsilon_i$ is a conjunction of equalities an inequalities 
of a complete equality type on $\textbf{cxy}$.
That is, $\varepsilon_i$ is  a maximal satisfiable set of 
equalities and inequalities involving the constants $\textbf{c}$ and variables $\textbf{xy}$. 

In general, some of the negated atomic formulas and inequalities in $\phi'$ 
may not be guarded. Let $\phi''$ be obtained from $\phi'$ by removing all conjuncts 
that are unguarded negative atomic formulas or unguarded inequalities.  

We claim that $\phi'$ and $\phi''$ are equivalent.
One direction is obvious, since $\phi'$ clearly implies $\phi''$. 
In the remainder of the proof, we show that $\phi''$  implies $\phi'$.

Consider an arbitrary structure $\frakA$ and tuple $\textbf{a}$ 
such that $\frakA\models\phi''(\textbf{a})$. It is our task to show 
that $\frakA\models\phi'(\textbf{a})$. Our general approach will be 
to construct another structure $\frakA'$ and tuple $\textbf{b}$ 
such that $\frakA'\models\phi'(\textbf{b})$. 
In addition, we will show that $(\frakA',\textbf{b})\to^s_{GN} (\frakA,\textbf{a})$. 
By Theorem~\ref{thm:strongbisim}, this will allow us to conclude 
$\frakA \models \phi'(\textbf{a})$ as needed, since $\phi'$ is logically 
equivalent to $\phi \in \gnfo$.
 
Let $h$ be a variable assignment from an appropriate
$\phi''_i(\textbf{x}) = \exists \textbf{y} \left(\varepsilon''_i \wedge \bigwedge_j \psi''_{ij}\right)$ 
to elements of $\frakA$, 
witnessing $\frakA \models \phi''(\textbf{a})$. 
In particular, $\varepsilon''_i$ is in general an incomplete equality type on $\textbf{cxy}$ 
that only includes an equality or inequality of every pair of variables 
that co-occur in a positive relational atom in some $\psi''_{ij}$.
We need to show that $\frakA \models \phi'_i(h(\textbf{x}))$. 
The main obstacles to overcome are:
\begin{enumerate}
\item[(i)] the possibility that $h$ maps two variables $u,v$ 
  to the same element of $\frakA$ while $\varepsilon'_i$ includes 
  the (unguarded) inequality $u\neq v$.
\item[(ii)] the possibility that $\frakA$ contains a fact that is the
  $h$-image of an atomic formula  occurring under an (unguarded)
  negation in $\phi'_i$.
\end{enumerate}
Based on these considerations, our construction of $\frakA'$ and $\textbf{b}$ will,
intuitively, involve (i) making sure that only those equalities are satisfied 
that are either explicitly contained in $\phi'_i$ or that follow (by transitivity) 
from guarded equalities true in $\frakA$ at $\textbf{a}$ and (ii) making sure 
that every fact satisfied in $\frakA'$ whose values are in the range of $h$ 
is guarded by a fact that is an $h$-image of a positive atomic formula of $\phi'_i$. 

The precise construction is as follows.
Let $X$ be the set of constants and all variables occurring, free or bound, in $\phi'_i$. 
Further let $\equiv$ be the equivalence relation on $X$ generated by all pairs of 
constants or variables $(u, v)$ such that $\varepsilon''_i$ contains the equality $u=v$. 
Let $f:X\to X/_\equiv$ be the natural map that sends each variable to its equivalence class.
We define the structure $\frakA^*$ with domain $X/_\equiv$ and, for each relation symbol $R$,
the relation $R^{\frakA^*}$ consisting of tuples $f(\textbf{u})$ such that  $R(\textbf{u})$ 
occurs as a positive atomic sub-formula in $\phi''_i$ or, what is the same, in $\phi'_i$.
Further let the $\equiv$-class of each constant interpret in $\frakA^*$ the corresponding 
constant symbol and let $\textbf{b}=f(\textbf{x})$.  
Note that $\frakA^*$ depends on $\frakA$ solely through the choice of the disjunct $\phi'_i$ 
that is assumed to be satisfied at $\textbf{a}$ in $\frakA$ via the variable assignment $h$.

\begin{itemize}
\item Observation 1: there is a homomorphism $g:\frakA^* \to dom(\frakA)$ 
  such that $h=g \circ f$ and such that $g$ is injective on guarded subsets of $\frakA^*$.
  That is, $g$  maps distinct elements co-occurring in a fact of $\frakA^*$ to distinct elements of $\frakA$.
\item Observation 2: $f$ assigns elements of $\frakA^*$ to variables of $\phi'$ 
      in a manner witnessing $\frakA^*\models\phi'(\textbf{b})$.  
\end{itemize}
Observation 1 follows from the definition of $\equiv$ and of $\frakA^*$. 
Observation 2 follows from the construction of $\frakA^*$ 
(for the equalities, inequalities, and positive atomic formulas) 
and from the previous observation (for the negative atomic formulas).

As a next step, we transform $\frakA^*$ into $\frakA'$ as follows. 
For each fact $F$ of $\frakA^*$ we make an isomorphic copy of $\frakA$ denoted $\frakA'_F$, 
where the isomorphism maps the elements belonging to 
the $g$-image of $F$ to their, by Observation 1, unique $g$-preimage 
and maps all other elements to distinct fresh elements. 
We define $\frakA'$ as the union 
  $\frakA^* \cup \bigcup \{ \frakA'_F \mid F \text{ a fact of } \frakA^* \}$, 
and let $\widehat{g}:\frakA'\to \frakA$ be the map that extends $g$ 
by mapping every newly-created element in some $\frakA'_F$ to the corresponding 
element of $\frakA$. Note that, by construction, $\widehat{g}: \frakA^* \to \frakA$ is a homomorphism.

\begin{itemize}
\item Observation 3: $\frakA'\models\phi'_i(\textbf{b})$ via the variable assignment $f$. 
\item Observation 4: $(\frakA',\textbf{b})\to^s_{GN} (\frakA,\textbf{a})$.
\end{itemize}

Observation 3 follows from Observation 2,  $\frakA^* \subinst \frakA'$,
and the observation that $\frakA'$ does not add any new facts on elements of $\frakA^*$.
For Observation 4, it can be easily verified that the graph of $\widehat{g}$ 
is in fact a strong GN-bisimulation, which is compatible with the homomorphism $g$ 
and $g(\textbf{b})=\textbf{a}$. 
From Observation 4 and Theorem~\ref{thm:strongbisim} 
we get that $\frakA\models\phi'(\textbf{a})$ as needed.
\end{proof}

\noindent
{\bf Note.} This theorem can also be proven by refining the \gnfo interpolation 
theorem of Section~\ref{sec:interpol} to get a Lyndon-style interpolation theorem. 
The approach via interpolation is spelled out in the paper~\cite{csllics14}.

\medskip

Finally, we consider the situation for  \gnfo formulas that are positive existential 
(for short, $\exists^+$). Since \gnfo contains all $\exists^+$ formulas, 
Rossman's homomorphism preservation theorem~\cite{ross} implies that the $\exists^+$ formulas are exactly 
the formulas in \gnfo preserved under homomorphism, over all structures or 
(equivalently, by the finite model property for \gnfo) over finite structures.
In addition, using the proof of Rossman's theorem plus the decidability of \gnfo 
we can effectively decide whether a \gnfo formula can be rewritten 
in $\exists^+$.

\begin{theorem} \label{thm:rossman}
There is an effective algorithm for testing whether a given \gnfo
formula is equivalent to a positive existential
formula, and, if so, computing such a formula.
\end{theorem}

\begin{proof}
Rossman's proof \cite{ross} shows that if an arbitrary FO formula $\phi$ is
equivalent to an $\exists^+$ formula, it is equivalent to one
of the same quantifier rank as $\phi$. If $\phi$ is in \gnfo, we can
test equivalence of a given $\exists^+$ formula $\phi'$ with $\phi$, using
the decidability of \gnfo. We can thus test all  $\exists^+$ formulas
with quantifier rank bounded by the quantifier rank of $\phi$, giving an
effective procedure.
\end{proof}

\section{Interpolation and Beth definability for $\gnfo$} \label{sec:interpol}
The Craig Interpolation theorem for first-order logic~\cite{craig57beth} can be stated as follows:
given formulas $\phi, \psi$ such that $\phi\models \psi$,
there is a formula $\chi$ such that 
\begin{enumerate}
\item[(i)] $\phi \models \chi$, and $\chi \models \psi$
\item[(ii)] all relations occurring in $\chi$ occur in both $\phi$ and $\psi$
\item[(iii)] all constants occurring in $\chi$ occur in both $\phi$ and $\psi$
\item[(iv)] all free variables of $\chi$ are free variables of both $\phi$ and $\psi$.
\end{enumerate}

The Craig Interpolation theorem has a number of important consequences, including
 the \emph{Projective Beth Definability theorem}~\cite{beth}.
Suppose that we have a sentence $\phi$ over a first-order signature
of the form $\sigma\cup  \{G\}$, where $G$ is an $n$-ary predicate,
 and suppose $\sigma'$ is a subset of $\sigma$.
A sentence $\phi$ \emph{implicitly defines predicate $G$ over $\sigma'$} if: 
for every $\sigma'$-structure $I$, every
 expansion 
to a $\sigma\cup  \{G\}$-structure  $I'$ satisfying $\phi$ has the same
restriction to $G$.
Informally, the $\sigma'$ structure and the sentence $\phi$ determine a unique value for $G$.
An $n$-ary predicate $G$  is \emph{explicitly  definable  over $\sigma'$ for models of $\phi$}
if there is another formula  $\rho(x_1 \ldots x_n)$ using only predicates from $S'$ such
that $\phi \models \forall \vec x ~ \rho(\vec x) \leftrightarrow G(\vec x)$.
It is easy to see that whenever $G$ is explicitly definable over $\sigma'$ 
for models of $\phi$, then $\phi$ implicitly defines $G$ over $\sigma'$. 
The Projective Beth Definability theorem states the converse: 
 if $\phi$ implicitly defines $G$ over $\sigma'$, then $G$ is explicitly definable 
 over $\sigma'$ for models of $\phi$.
In the special case where $\sigma'=\sigma$, this is called simply the Beth
Definability theorem. \looseness=-1

A proof of the  Craig Interpolation theorem can be found in any model theory
textbook (e.g.~\cite{ChangKeisler}). 
The Projective Beth Definability theorem follows from the Craig Interpolation theorem.
Both theorems fail when restricted to finite structures \cite{EF99}.

We say that a fragment of first-order logic has the Craig Interpolation Property (CIP) 
if for all $\phi \models \psi$ in the fragment, the result above holds relative to the fragment. 
We similarly say that a fragment satisfies the Projective Beth Definability Property (PBDP) 
if the  Projective Beth Definability theorem holds relativized to the fragment 
-- that is, if $\phi$  in the hypothesis of the theorem lies in the fragment
then there is a corresponding formula $\rho$  lying in the fragment as well. 
We talk about the Beth Definability Property (BDP) for a fragment in the same way.
The argument for first-order logic applies to any fragment with reasonable
closure properties \cite{hooglandthesis} to show that CIP implies PBDP.

CIP and PBDP do not hold when implication is restricted to finite models \cite{EF99}.
However,  the finite and unrestricted versions of these properties are equivalent
when considering fragments of FO with some basic closure
properties that have the finite model property, since
there equivalence (resp. consequence)
over finite structures can be replaced by equivalence (resp. consequence) over all structures. 
Thus it is particularly natural to look at CIP and PBDP for such fragments, such as $\gf$ and $\gnf$.
Hoogland,  Marx, and Otto \cite{HMO}  showed that
the Guarded Fragment satisfies BDP but lacks CIP. 
Marx \cite{Marx07pods} 
went on to explore PBDP for the Guarded Fragment and its extensions.
He argues that
the PBDP holds for  an extension of $\gf$ called the Packed Fragment.
The definition of the Packed Fragment is
not important for this work, but 
at the end of this section we  show that PBDP fails for $\gf$, and also (contrary to \cite{Marx07pods}) for the Packed
Fragment.  
But we will adapt ideas of Marx to show that CIP and PBDP do hold for $\gnf$. 

\medskip

The main technical result of this section is then:


\begin{theorem}[$\gnf$ has Craig interpolation]
\label{thrm:Craig_with_constants}
For each pair of \gnfo-formulas $\phi, \psi$ such that $\phi\models \psi$,
there is a \gnfo-formula $\chi$ such that
\begin{compactitem}
\item[(i)] $\phi \models \chi$, and $\chi \models \psi$,
\item[(ii)] all relations occurring in $\chi$ occur in both $\phi$ and $\psi$,
\item[(iii)] all free variables of $\chi$ are free variables of both $\phi$ and $\psi$.
\end{compactitem}
\end{theorem}

Section~\ref{sec:proofcraig} is dedicated to the proof of Theorem~\ref{thrm:Craig_with_constants}.
In Section~\ref{sec:applicationsCraig} we present further
applications of the result, and in Section~\ref{sec:gf-interpolation} we discuss failure of 
interpolation for the Guarded Fragment.

We first comment that item (iii) can be ensured by pre-processing
$\phi$ and $\psi$.
We can assume that $\phi$ contains only free variables that are common
to $\psi$: if it has variables that are not,  then
we can existentially quantify them.
We can also assume that $\psi$ has only free variables that
are common to $\phi$: if it has variables that are not, then
we can universally quantify them, restricting the universal
quantification to a new ``dummy guard''. This new guard
will not  occur in the interpolant, since it is not common, so
this does not impact the other items.
quantifying any violating free variables of the interpolant. 
Thus
it suffices to ensure (i) and (ii).

Also observe that in Theorem~\ref{thrm:Craig_with_constants}, the
interpolant is allowed to contain constant symbols outside of the
common language.  Indeed, this must be so, for $\gnf$ lacks the stronger
version of interpolation where the interpolant can only contain
constant symbols occurring both in the antecedent and in the
consequent. Recall that, in \gnfo, as well as \gfo, constant symbols are
allowed to occur freely in formulas, and that their occurrence is not
governed by guardedness conditions. In particular, for example, the
formula $\forall y R(c,y)$ belongs to \gfo (and is equivalent to 
a formula of  \gnfo), while the
formula $\forall y R(x,y)$ does not. Now, consider the valid
entailment $(x=c)\land \forall y R(c,y) ~~\models~~ (x=d)\to\forall y
R(d,y)$. It is not hard to show that any interpolant $\phi(x)$ not
containing the constants $c$ and $d$ must be equivalent to
$\forall y R(x,y)$. This shows that there are valid \gfo-implications
for which interpolants cannot be found in \gnfo, if the interpolants
are required to contain only constant symbols occurring both in the
antecedent and the consequent. In fact, in \cite{tencate:JSL05} it was
shown that, in a precise sense, every extension of \gfo with
this strong form of interpolation has full first-order expressive
power and is undecidable for satisfiability.

\subsection{Proof of Craig interpolation for \gnfo} 
\label{sec:proofcraig}

To establish Theorem~\ref{thrm:Craig_with_constants} we follow 
a common approach in modal logic (see, in particular, 
Hoogland, Marx, and Otto~\cite{HMO}). We make use of a result saying 
that we can take two structures over different signatures, behaving similarly 
in the common signature, and \emph{amalgamate} them to get a structure 
that is simultaneously similar to both of them (in the respective signatures).
The precise statement of the theorem will be in terms of the notion of 
strong GN-bisimulation introduced in Section~\ref{sec:preserve}, 
and the proof will make use of the results there. 
Our specific amalgamation construction is inspired by the \emph{zig-zag products} 
introduced by Marx and Venema~\cite{MarxVenema}.
In the lemma and claims below, $\textbf{a}$
will range over tuples, not necessarily guarded.

\begin{lemma}[Amalgamation]\label{lem:amalgamation} ~\\
Let $\sigma$ and $\tau$ be signatures containing the same constant
symbols but possibly different relation symbols.
If $(\mathfrak{A},\textbf{a})\to^s_{GN[\sigma\cap\tau]}
(\mathfrak{B},\textbf{b})$, then there is a structure
$(\mathfrak{U},\textbf{u})$ such that
$(\mathfrak{A},\textbf{a})\to^s_{GN[\sigma]}
(\mathfrak{U},\textbf{u})\to^s_{GN[\tau]} (\mathfrak{B},\textbf{b})$
\end{lemma}

\begin{proof}
     Let $Z$ be the strong GN-bisimulation between $\mathfrak{A}$ and
     $\mathfrak{B}$ witnessing the fact that
     $(\mathfrak{A},\textbf{a})\to^s_{GN[\sigma\cap\tau]}
     (\mathfrak{B},\textbf{b})$.  Below, for any
     partial 
  map $f$ from $\mathfrak{A}$ to $\mathfrak{B}$ or vice versa, with a
  slight abuse of notation, we will
  write $f\in Z$ if $f$ can be extended to a homomorphism that is
  compatible  with $Z$. 
  In particular, we have $(\textbf{a} \mapsto \textbf{b})\in Z$. 
    Note that, for individual elements $c$ and $d$,  $(c\mapsto d)\in Z$ if and only if
   $(d\mapsto c)\in Z$.
   In addition, with some further abuse of notation, for any $k$-tuple
   $\textbf{c}=c_1\ldots c_k$ of elements of $\mathfrak{A}$ 
   and  for any $k$-tuple $\textbf{d}=d_1\ldots d_k$ of elements of
   $\mathfrak{B}$,
  we will denote by $\langle\textbf{c},\textbf{d}\rangle$ the $k$-tuple
   $((c_1,d_1), \ldots, (c_k,d_k))$. 

  We define the amalgam $(\mathfrak{U},\textbf{u})$ as follows: 
  \begin{itemize}[--]
    \item the domain of $\mathfrak{U}$ is 
         $\{ (c,d)\in \mathfrak{A}\times \mathfrak{B} \mid (c\mapsto d) \in Z \}$;
    \item
     $R^\mathfrak{U} = \{ \langle \textbf{c},\textbf{d}\rangle \mid
         \textbf{c}\in R^\mathfrak{A} \text{ and } (\textbf{c}\mapsto
         \textbf{d})\in Z\}$ for every $R\in\sigma$;
    \item
     $S^\mathfrak{U} = \{ \langle \textbf{c},\textbf{d}\rangle\mid
         \textbf{d}\in S^\mathfrak{B} \text{ and } (\textbf{d}\mapsto
         \textbf{c})\in Z\}$ for every $S\in\tau$;
    \item
      $c^\mathfrak{U} = (c^\mathfrak{A},c^\mathfrak{B})$ for every constant symbol $c$;
    \item 
      $\textbf{u} = \langle\textbf{a},\textbf{b}\rangle$.
  \end{itemize}
  To see that $\mathfrak{U}$ is thus well defined, note that for $R\in\sigma\cap\tau$, 
  if $\textbf{c}\in R^\mathfrak{A}$ and $(\textbf{c}\mapsto \textbf{d})\in Z$ 
  then also $\textbf{d}\in R^\mathfrak{B}$ and $(\textbf{d}\mapsto \textbf{c})\in Z$, 
  and vice versa. 

   \begin{trivlist}
   \item \textbf{Claim 1: $(\mathfrak{A},\textbf{a})\to^s_{GN[\sigma]} (\mathfrak{U},\textbf{u})$}
   \item \textit{Proof of claim 1.} 
   Let $Z'$ be the collection of all pairs 
     $(\textbf{v},\langle\textbf{v},\textbf{w}\rangle)$ 
   for $(\textbf{v}\mapsto\textbf{w})\in Z$ and $\textbf{v}$ guarded
   (by a $\sigma$-atomic formula) in $\mathfrak{A}$. 
   We will show that $Z'$ is a strong GN-bisimulation 
   between $\mathfrak{A}$ and $\mathfrak{U}$, 
   and that $(\mathbf{a}\mapsto\mathbf{u})\in Z'$.

   Consider any pair $(\textbf{v},\langle\textbf{v},\textbf{w}\rangle)\in Z'$. 
   By construction, we have that $(\textbf{v},\textbf{w})\in Z$ and hence, 
   there is a homomorphism $h:\mathfrak{A}\to\mathfrak{B}$ 
   that is compatible with $Z$, and such that $h(\textbf{v})=\textbf{w}$.  
   Let $\widehat{h}(a)=(a,h(a))$ for all $a \in \frakA$. 
   It can easily be verified that $\widehat{h}$ is a homomorphism 
   from $\mathfrak{A}$ to $\mathfrak{U}$ that is compatible with $Z'$, 
   and that $\widehat{h}(\textbf{v})=\langle \textbf{v},\textbf{w}\rangle$. 
   Conversely, we also need to show that there is a homomorphism
   from $\mathfrak{U}$ to $\mathfrak{A}$ that is compatible with $Z'$
   and that maps $\langle \textbf{v},\textbf{w}\rangle$ to $\textbf{v}$. 
   Here, we can simply choose the natural projection as our homomorphism. 
   It is easy to verify that this satisfies the requirements.

   Finally, we need to show that $(\mathbf{a}\mapsto\mathbf{u})\in Z'$, 
   i.e., that there is a homomorphism from $\mathfrak{A}$ to $\mathfrak{B}$ 
   that is compatible with $Z'$ and that sends $\mathbf{a}$ to $\mathbf{u}$. 
   Recall that $\mathbf{u}=\langle\mathbf{a},\mathbf{b}\rangle$. 
   Let $h$ be a homomorphism from $\mathfrak{A}$ to $\mathfrak{B}$
   that is compatible with $Z$ and that sends $\mathbf{a}$ to $\mathbf{b}$, 
   and let $\widehat{h}$ be defined by $\widehat{h}(a)=(a,h(a))$ for all $a \in \frakA$. 
   It is easy to verify that $\widehat{h}$ satisfies the requirements.
   \hfill $\dashv$

   \item \textbf{Claim 2: 
         $(\mathfrak{U},\textbf{u})\to^s_{GN[\tau]} (\mathfrak{B},\textbf{b})$ }
   \item \textit{Proof of claim 2.} 
       the relevant strong GN-bisimulation $Z''$ is constructed analogously  to $Z'$ above.
       Note that, in this case, we do not get that 
       $(\textbf{b}\mapsto\textbf{u})\in Z''$ but we get that
       $(\textbf{u}\mapsto \textbf{b})\in Z''$ 
       because this partial map is included in the natural projection 
       from $\mathfrak{U}$ to $\mathfrak{B}$, which is compatible with $Z''$. 
       \hfill $\dashv$
  \end{trivlist}
\end{proof}


\begin{proof}[Proof of Theorem~\ref{thrm:Craig_with_constants}]
As mentioned
earlier, without loss of generality we can assume that $\phi$ and $\psi$ have the same 
free variables. We can also assume
 they  reference the same set of constant symbols (eg.~by appending 
vacuous identities  $c_j=c_j$ as conjuncts to either formula as needed).
With this proviso let $\phi(\textbf{x})$ and $\psi(\textbf{x})$ be $\gnf$-formulas 
with free variables $\textbf{x}$ such that 
  $\models\forall\textbf{x}(\phi(\textbf{x})\to\psi(\textbf{x}))$; 
let $\sigma$ and $\tau$ denote their respective signatures and suppose, 
for the sake of contradiction, that there is no $\gnf[\sigma\cap\tau]$-interpolant. 

As a first step, using a standard compactness argument, we establish the existence 
of two structures $(\mathfrak{A},\textbf{a})$ and $(\mathfrak{B},\textbf{b})$ 
such that $\mathfrak{A}\models\phi(\textbf{a})$, $\mathfrak{B}\models\neg\psi(\textbf{b})$, 
and $(\mathfrak{A},\textbf{a})\Rrightarrow_{GN[\sigma\cap\tau]} (\mathfrak{B},\textbf{b})$. 

We now argue for this first step. Let $\Phi(\textbf{x})$ be the set of 
all $\gnf[\sigma\cap\tau]$ consequences of $\phi(\textbf{x})$ using
only free variables in $\textbf{x}$.
By the assumption that there is no interpolant
and  compactness, we know that $\Phi(\textbf{x})$ cannot imply $\psi(\textbf{x})$.
Therefore, there is a structure $\mathfrak{B}\models\Phi(\textbf{b})\land\neg\psi(\textbf{b})$. 
Next, consider
$$
  \Psi(\textbf{x}) = \{ \neg \eta(\textbf{x}) \mid 
      \eta(\textbf{x}) \in \gnf[\sigma\cap\tau],\ \mathfrak{B} \models \neg\eta(\textbf{b}) \}
$$
and notice that $\Psi(\textbf{x})$ does not imply $\neg\phi(\textbf{x})$. 
For otherwise there would be, due to compactness, some natural number $k$ 
and $\neg\eta_0(\textbf{x}), \ldots, \neg\eta_{k-1}(\textbf{x}) \in \Psi(\textbf{x})$
such that $\bigwedge_{j<k} \neg\eta_j(\textbf{x}) \models \neg\phi(\textbf{x})$ 
ie.~$\phi(\textbf{x}) \models \bigvee_{j<k} \eta_j(\textbf{x})$ 
and thus $\bigvee_{j<k} \eta_j(\textbf{x}) \in \Phi(\textbf{x})$, 
because $\bigvee_{j<k} \eta_j(\textbf{x}) \in \gnf[\sigma\cap\tau]$, 
implying $\mathfrak{B} \models \bigvee_{j<k} \eta_j(\textbf{b})$ 
in contradiction to the fact that $\eta_j(\textbf{x}) \in \Psi(\textbf{x})$ 
and hence $\mathfrak{B} \models \neg\eta_j(\textbf{b})$ for each $j<k$.
Therefore, there is a structure $\mathfrak{A}\models\Psi(\textbf{a})\land\phi(\textbf{a})$. 
By construction, we have that 
$(\mathfrak{A},\textbf{a})\Rrightarrow_{GN[\sigma\cap\tau]}(\mathfrak{B},\textbf{b})$.

Note that in the above step we can ensure that both $\frakA$ and $\frakB$ are countable. 
Thus, using Lemma~\ref{lem:lifting}, we can lift the $\Rrightarrow_{GN[\sigma\cap\tau]}$ 
relationship between $(\mathfrak{A},\textbf{a})$ and $(\mathfrak{B},\textbf{b})$ 
to a $\to^s_{GN[\sigma\cap\tau]}$ relationship between respective elementary extensions
$(\widehat{\mathfrak{A}},\textbf{a})$ and $(\widehat{\mathfrak{B}},\textbf{b})$. 
Applying the Amalgamation Lemma~\ref{lem:amalgamation} to these extensions 
we obtain $(\mathfrak{U},\textbf{u})$ such that 
$
  (\widehat{\mathfrak{A}},\textbf{a}) \to^s_{GN[\sigma]}
  (\mathfrak{U},\textbf{u}) \to^s_{GN[\tau]}  
  (\widehat{\mathfrak{B}},\textbf{b}) \, .
$
Observe that $\mathfrak{U}\models\phi(\textbf{u})$ follows 
from $\widehat{\mathfrak{A}} \models \phi(\textbf{a})$ 
and $(\widehat{\mathfrak{A}},\textbf{a}) \to^s_{GN[\sigma]}$.
Similarly, we can infer $\mathfrak{U} \models \neg\psi(\textbf{u})$
for otherwise $(\mathfrak{U},\textbf{u}) \to^s_{GN[\tau]} (\widehat{\mathfrak{B}},\textbf{b})$
would allow us to conclude $\widehat{\mathfrak{B}} \models \psi(\textbf{b})$
contradicting our choice of $(\widehat{\mathfrak{B}},\textbf{b})$.
Thus we have found $\mathfrak{U}\models\phi(\textbf{u})\land\neg\psi(\textbf{u})$
contradicting the assumption that $\phi(\textbf{x})$ implies $\psi(\textbf{x})$.
\end{proof}



\subsection{Applications of Interpolation}
\label{sec:applicationsCraig}

An analogue of the Projective Beth Definability theorem~\cite{beth}
for $\gnf$ follows from Craig interpolation by standard arguments \cite{hooglandthesis}.

\begin{corollary} \label{cor:Beth_with_constants}
If a \gnfo sentence $\phi$ in signature $\sigma$ implicitly defines
a relation symbol $R$ in terms of a signature $\tau \subset \sigma$,
and $\tau$ includes all constants from $\sigma$, then there is
an explicit definition of $R$ in terms of $\tau$ relative to $\phi$.
\end{corollary}


We now investigate properties pertaining to ``view-based query rewriting''
for $\gnf$.
Suppose $V$ is a finite set of relation names, 
and we have FO formulas $\{\phi_v:v \in V\}$ over a signature $\sigma$ 
that is disjoint from $V$.
Suppose $\phi_Q$ is another  first-order formula over the signature $\sigma$. 
The family of formulas $\{ \phi_v:v \in V \}$ \emph{determine $\phi_Q$ over finite structures} 
if for all finite $\sigma$-structures $I$ and $I'$ with $\phi_v(I)=\phi_v(I')$ for all $v \in V$, 
we have $\phi_Q(I)=\phi_Q(I')$. Similarly, we say that the set $\{\phi_v:v \in V\}$  determine $\phi_Q$ 
over all structures if the above holds for all $I$ and $I'$. Unwinding the definitions, the
reader can see that the latter assertion is the same
as stating that the sentences asserting
\[
\forall \vec x ~ \phi_v(\vec x) \leftrightarrow v(\vec x)
\]
for each $v \in V$
as well as
\[
\forall \vec x ~ \phi_Q(\vec x) \leftrightarrow Q(\vec x)
\]
implicitly define the relation $Q$ over the signature $V$.
In the database literature, the symbols $v \in V$ are often referred to
as ``view relations'' and the corresponding formula $\phi_v$ is the ``view definition for $v$''.

From the  PBDP we know that when $\{\phi_v:v \in V\}$  determine $\phi_Q$
over all structures, there is a first-order formula $\rho$ over $V$
that explicitly defines $Q$.
Such a $\rho$ is called
a \emph{rewriting of  $\phi_Q$  over $\{\phi_v:v \in V\}$}.
Segoufin and Vianu initiated a study of determinacy for special classes
of formulas $\phi_v$ and $\phi_Q$,
including the question of deciding when determinacy and determinacy-over-finite-structures holds, and
examining
when the assumption of determinacy implies 
that the rewriting is realized  by a formula in a restricted logic.
Nash, Segoufin, and Vianu showed that determinacy over finite structures for unions of conjunctive queries 
is undecidable~\cite{NSV10TDBS}, and that for UCQs determinacy over finite structures 
does not imply rewritability even in first-order logic. More recently determinacy
for conjunctive queries has been shown undecidable both over finite
structures and over all structures \cite{redspider,rainworm}. The fact that  determinacy of FO queries does not imply FO rewritability
over finite structures is related to the fact that CIP, PBDP, and BDP all fail
for FO when implication is considered over finite structures.

We will use the PBDP above to show  that whenever 
$\{\phi_v: v \in V\}$ determines  $\phi_Q$ 
and additionally both  $\{\phi_v: v \in V\}$ and $\phi_Q$ 
are answer-guarded $\gnfo$ formulas,
then there is  a first-order rewriting, and even a rewriting  in GNFO.
 Recall from Section \ref{sec:prelim}
that answer-guarded formulas are those of the form
$\phi(\vec x)=R(\vec x) \wedge \phi'$ for some $\phi'$ and relation symbol $R$. 

 Note that rewritings of determined queries, when they  exist, 
can always  be taken to be domain-independent queries, since $\phi_Q(I)$ 
is,  by definition of  determinacy, only dependent on $\phi_v(I)$ for $v \in V$.
Observe also that if we have 
then determinacy of  formula
$\phi_Q$ by a family
of formulas $\{\phi_v: v \in V\}$ can be expressed as validity of a  sentence 
with a vocabulary suitable for talking about two structures of the original
signature. The sentence is:
\begin{align*}
[\bigwedge_{v \in V} \forall \vec x ~ (\phi_v(\vec x) \leftrightarrow 
\phi'_v(\vec x) )] 
\wedge \phi_Q(\vec c) \\
\rightarrow  \phi'_Q(\vec c)
\end{align*}
where $\vec c$ is a set of fresh constants, $\phi'_v$
is formed from $\phi_v$ by replacing each relation
$R$ by a copy $R'$, and $\phi'_Q$ is similarly formed from
$\phi_Q$.
If $\phi_Q$ is in $\gnfo$ and each $\phi_v$ is an
answer-guarded $\gnfo$ formula, then this sentence is in $\gnfo$.
Thus from the finite model property of $\gnfo$,
when $\phi_Q$ is in $\gnfo$ and each $\phi_v$ is an
answer-guarded $\gnfo$ formula, determinacy over finite
structures implies determinacy   over all structures.
Similarly,  Theorem \ref{thm:gnfsat} implies that
``$\{\phi_v: v \in V\}$ determine $\phi_Q$'' can be
decided in $\twoexptime$, when the $\phi_v$ range over
answer-guarded $\gnf$ formulas and $\phi_Q$ ranges
 over 
$\gnf$ formulas.

We can now state the consequence of the PBDP for
determinacy-and-rewriting (relying again on the finite model
property of \gnf).

\begin{corollary} \label{cor:detrew} Suppose a set of answer-guarded $\gnf$ queries $\{\phi_v: v \in V\}$ 
determines
an answer-guarded $\gnf$ query $\phi_Q$ over finite structures. 
Then there is a $\gnf$ query $\rho$ that is a rewriting.
Furthermore, there is an algorithm that, given
$\phi_v$'s and $\phi_Q$ satisfying the hypothesis, effectively
finds such a formula~$\rho$.
\end{corollary}

\begin{proof}
Extend the vocabulary with predicates $v$ for each $\phi_v$ and a predicate $Q$
for $\phi_Q$. Now consider a sentence stating that each $v$ contains  exactly the tuples
satisfying $\phi_v$ and that $Q$ contains exactly the tuples satisfying
$\phi_Q$. The hypotheses imply that this sentence is in $\gnf$, and that it implicitly
defines $Q$ with respect to the signature containing only the symbols in $V$, when
restricting to finite structures. Using the finite model property
of $\gnf$, we see that implicit definability
hold over all structures.
Applying the PBDP for $\gnf$, we get an explicit definition of $Q$ in $\gnf$.
By unwinding the definitions we see that this is a rewriting.

The rewriting can be found effectively by simply enumerating every possible $\rho$ and checking
whether $\phi_Q$ is logically equivalent to $\rho(V_1/\phi_1 \ldots V_n/\phi_n))$; 
the check is effective using the decidability of equivalence for $\gnf$ \cite{BtCS15jacm}.
\end{proof}

Work subsequent to this article has obtained tight bounds on the rewritings
\cite{csllics14}, via a constructive approach to $\gnfo$ interpolation. 

\medskip

Recall from our discussion above that rewritings are domain-independent, 
since they depend only on the facts produced by the view definitions.
Thus, as discussed in Section~\ref{sec:prelim}, they can be converted to $\gnra$.
Note also that  $\gnf$ views $V$ can check properties of a structure
(e.g.~linear TGDs) 
as well as return results. Using the above, we can get the following 
variant of Corollary~\ref{cor:detrew} for sentences and queries:

Suppose a set of  answer-guarded UCQ views $\{\phi_v: v \in V\}$ determine
an answer-guarded UCQ $\phi_Q$ on finite structures satisfying a set of $\gnf$
sentences $\Sigma$.
Then there is a  $\gnf$ rewriting of $Q$ using $V$ that is valid
over structures satisfying $\Sigma$.


\subsection{Negative results for the Guarded Fragment and packed fragments}
\label{sec:gf-interpolation}
We  now prove that PBDP
 fails for the Guarded Fragment.
This suggests, intuitively, that if we want to express explicit definitions
even for $\gf$ implicitly-definable relations, we will need to use 
all of $\gnfo$.

\begin{theorem} \label{thm:failpbdpgnfok} 
The PBDP fails for $\gf$.
\end{theorem}

\vspace{-3mm}
\begin{proof}
Consider the GF sentence $\phi$ that
is the conjunction of  the following:
\[\begin{array}{rclr} 
    \forall x ~[C(x) &\rightarrow&  \exists yzu ~ (G(x,y,z,u)\land E(x,y)
    \land E(y,z) \land E(z,u)\land E(u,x))] \\   
    \forall x y ~ [ ~ (E(x,y) \wedge \neg C(x) ) &\rightarrow& P_0(x) \land \neg P_1(x) \land
    \neg P_2(x)]\\
    \forall x y ~ [ ~ (P_i(x) \wedge E(x,y)) &\rightarrow& P_{(i+1\text{ mod
      } 3)}(y)] \text{~~~
      for all $0\leq i< 3$}
\end{array}\]
The first sentence forces that if $C(x)$ holds, then $x$ lies on a directed $E$-cycle of length $4$.
The remaining two sentences force that if $\neg C(x)$ holds, then $x$ only 
lies on directed $E$-cycles
 whose length is a multiple of $3$.
Clearly, the relation $C$ is implicitly defined in terms of $E$.

However, we claim there is no explicit definition in $\gf$ in terms of $E$, because
no formula of $\gf$ can distinguish the directed $E$-cycle of length $k$ from the
directed $E$-cycle of length $\ell$ for  $3\leq k<\ell$.
Here we will  make use of the
notion of guarded bisimulation between  
structures $\frakA$ and $\frakB$, due to  Andr\'eka, van Benthem,  and  N\'emeti\cite{AvBN98JPL}.
This  is a non-empty family of 
partial isomorphisms from $\frakA$ to $\frakB$
satisfying
the following back-and-forth conditions: \begin{itemize}  \item For every partial isomorphism
$f \in I$ with domain $X$ and every guarded subset $X'$ of the domain of $\frakA $, there
 is a partial isomorphism $g \in I$ whose domain contains $X'$ agreeing with $f$ on $X \cap X'$
\item  for $f \in I$ with co-domain $Y$ and every guarded subset $Y'$ of the domain of $\frakB$, there
 is a partial isomorphism $g \in I$ with domain containing $Y'$ such that $g^{-1}$ and $f^{-1}$ agree on $Y \cap Y'$
\end{itemize}
It is known \cite{AvBN98JPL} that if two structures
are guarded bisimilar, then they must agree on all sentences
of $\gfo$.

Fix a binary relation symbol $E$, let $C_k$ be the directed $E$-cycle
of length $k$. Let $3\leq k, \ell$, and let
$Z$ be the binary relation containing all pairs $((a,b),(c,d))$ such
that $(a,b)\in E^{C_k}$ and $(c,d)\in E^{C_\ell}$. One can verify directly that
$Z$ is a guarded-bisimulation between $C_k$ and
$C_\ell$.
\end{proof}

It follows from Theorem~\ref{thm:failpbdpgnfok} that \gf lacks CIP as well, which was
already known \cite{HMO}. Furthermore, 
the above argument can be adapted to show that determinacy does not imply 
rewritability for views and queries defined in \gf: 
consider the set of views $\{\phi_{v_1}, \phi_{v_2}\}$, where
$\phi_{v_1}=\phi$ and $\phi_{v_2}(x,y) = E(x,y)$. Clearly,
$\{\phi_{v_1}, \phi_{v_2}\}$ determine the query $Q(x) = \phi\land C(x)$. 
On the other hand, any rewriting would constitute an explicit
definition in \gf of $C$ in terms of $E$, relative to $\phi$, which we know does not exist.

In \cite[Lemma 4.4]{Marx07pods} it was asserted that PBDP holds for an extension of the
Guarded Fragment, called the \emph{Packed Fragment}, in which a guard
$R(\vec{x})$ may be a conjunction of atomic formulas, as long as every
pair of variables from $\vec{x}$ co-occurs in one of these conjuncts.

The proof of Theorem~\ref{thm:failpbdpgnfok}, however, shows that
PBDP fails for the Packed Fragment, because 
known results (cf.~\cite{Marx07pods}) imply that no
formula of the Packed Fragment can distinguish the cycle of length $k$
from the cycle of length $\ell$ for $4\leq k<\ell$.  
This can also be shown by appealing to the notion of
packed bisimulation \cite{Marx07pods}, a variant
of guarded bisimulation which characterizes
expressibility in
the Packed Fragment.
In fact the relation $Z$ defined in the proof of  Theorem~\ref{thm:failpbdpgnfok} is a
packed bisimulation between $C_k$ and $C_\ell$. This shows that no
sentence of the Packed Fragment can distinguish directed $E$-cycles of
different length. Incidentally, the  sentence $\exists
xyz ~ (Rxy\land Ryz\land Rzx)$ distinguishes $C_3$
from $C_4$. By writing it as $\exists xyz ~ (Rxy\land Ryz\land Rzx) \wedge \top$ we
see that this sentence is in the Packed Fragment.
Indeed, it turns out that
there is a flaw in the proof of Lemma 4.4 in \cite{Marx07pods}.

\section{Expressibility of certain answers for queries with respect to $\gnf$ TGDs} \label{sec:rewritedep}
We now turn to a different set of issues about rewriting
formulas into a certain syntax. These
questions will be motivated by issues in databases and knowledge 
representation, rather
than general model-theoretic concerns. Constructions on models
will  be utilized to prove the rewritability results, as in the previous
sections. But while the construction of the previous sections
were geared towards first-order logic and some traditional
subsets (e.g. positive existential formulas),
the constructions in the remainder of the paper
will be tailored to formulas having a 
more specialized syntax (TGDs).

A fundamental concept in the study of information integration and ontology-mediated data access is the notion of 
\emph{certain answers} for a conjunctive query with respect to a database instance and a collection of sentences. For the sake of consistency in
the presentation, we define certain answers here
in terms of structures, rather than database instances. Note that the queries and sentences
that we consider in this section are all domain independent. Hence, as pointed out
in Section~\ref{sec:prelim}, their evaluation is determined by the underlying
instance of a structure, and hence in this section we can make use of constructions
taking instances to instances.

Given two structures $\frakA, \frakB$ over the same signature $\tau$, 
recall  the notation $\frakA\subinst \frakB$, meaning
that the two structures agree on the interpretation of the constant symbols, and,
for every relation $R\in\tau$,  $R^\frakA \subseteq R^\frakB$. 
Let  $\frakA$ be a finite structure, $\Sigma$ a set of sentences in some logic, and
$Q(x_1 \ldots x_k)$  a formula in some logic.
A tuple $(a_1 \ldots a_k) \in dom(\frakA)^k$ is a \emph{certain answer of $Q$ with respect to $\frakA$ and $\Sigma$}
if $\frakB, a_1 \ldots a_k \models Q$ in every model $\frakB$ of $\Sigma$ such that 
$\frakA\subinst\frakB$. 
Determining
which tuples
are certain answers is a central problem in information integration 
and ontology-mediated data access.
Typically $\Sigma$ is referred to as a set of \emph{integrity constraints} (or just ``constraints'' below, for brevity),
while $Q$ is the \emph{query}. The structure  $\frakA$ represents
incomplete information about  a structure, and the sentences $\Sigma$ represent a
constraint on the completion. A certain answer to query $Q$ is a result which
is already determined by $\Sigma$ and the presence of the facts in $\frakA$.
In some cases one considers the ``finite model analog'' of the above definition:
requiring that $\frakB, a_1 \ldots a_k \models Q$ in every  \emph{finite} model $\frakB$ of $\Sigma$ with $\frakA\subinst\frakB$. 
For the constraints $\Sigma$ we consider, there will be no distinction between the finite and unrestricted version
of the problems.

One of the benefits
of $\gnf$ is that one can effectively determine the certain answers whenever
$Q$ and $\Sigma$ are expressed in $\gnf$, and thus in particular for every $\Sigma$
in  $\gnf$ and
conjunctive query $Q$ \cite{bbo}.  But one can do better for $\gnfo$ formulas that are also TGDs.
Recall from Subsection \ref{subsec:tgdchar}
that these are, up to equivalence, frontier-guarded TGDs: TGDs where
there is a guard containing all exported variables.
Baget et~al.~\cite{baget2010} proved that for every set of frontier-guarded 
dependencies $\Sigma$ and conjunctive query $Q$, the certain answers can be computed 
in polynomial time in $\frakA$.
However, one could hope for more than just being able to compute the certain answers
in polynomial time.
A conjunctive query $Q$ is \emph{first-order rewritable} under
sentences $\Sigma$ if there is a first-order formula $\phi$ such that on any
finite structure $\frakA$, the tuples that satisfy $\phi$ in $\frakA$ are exactly the certain answers to $Q$ on $\frakA$ under $\Sigma$.
Thus  a query is first-order rewritable with respect to $\Sigma$
if we can reduce finding the certain answers to ordinary
evaluation of a first-order formula (which can be done, for example, with a database management system).
Unfortunately, it is known that there are  frontier-guarded TGDs and conjunctive queries
such that the certain answers can not be determined by evaluating a first-order query.  Indeed, this is true even for \emph{guarded TGDs}: recall
from Subsection \ref{subsec:tgdchar} that these are TGDs
where there is a single atom in the body containing all variables
of the body. A CQ and  guarded TGD that is not first-order
rewritable is given in Example \ref{ex:certain} below.
We will now look
at ways of ``remedying'' this situation.

We will show that we can decide, given a set $\Sigma$ of frontier-guarded
TGDs and a conjunctive query $Q$, whether or not $Q$ is first-order rewritable.
In this process, we will show that the certain answers can be expressed in a ``nice'' fragment
of Datalog, where Datalog is the extension of conjunctive queries with a fixpoint mechanism
(see  Section \ref{sec:prelim}).
One natural target for rewriting is a \emph{Guarded Datalog program}.
This is a Datalog program such that for every rule, the body of the
rule contains an atom over the input signature which contains all the
variables in the rule. The example below shows
why a language like Guarded Datalog is a natural target.

\begin{example} \label{ex:certain}
Consider a signature with binary
relations $R(x,y)$ and $S(x,y)$ as well as unary relation $U(x)$.

Consider the guarded TGDs:

\begin{align*}
\forall x y ~ [R(x,y) \wedge U(y) \rightarrow U(x)] \\
\forall x ~ [U(x) \rightarrow \exists z ~ S(x,z)] \\
\forall x y ~  [S(x,y) \rightarrow T(x)]
\end{align*}

and the query $Q(x) = T(x)$. 

One can check that the certain answers of $Q$ under $\Sigma$ on any structure $\frakA$ are
identical to the output of $P$ on $\frakA$, where $P$ is the Datalog program with the following rules:

\begin{align*}
UReach(x) \datalogarrow U(x) \\
UReach(x) \datalogarrow \exists y ~ R(x,y) \wedge UReach(y) \\
Goal(x) \datalogarrow  UReach(x) \\
Goal(x) \datalogarrow T(x) \\
Goal(x) \datalogarrow S(x,y)
\end{align*}

Notice that $P$ is a Guarded Datalog program, since
the body of each rule is guarded.
\end{example}

We will follow
(and correct) the approach of 
Baget~et~al.~\cite{bagetconf}, who argued that the certain answers
of conjunctive queries under frontier-guarded TGDs  are rewritable in Datalog. 
For guarded TGDs, this result had been announced by Marnette \cite{marnette}. 
The proof of Baget~et~al.~\cite{bagettr} revolves around a ``bounded base lemma'' 
showing that whenever a set of facts is not closed under ``chasing'' with FGTGDs, 
there is a small subset that is not closed (Lemma 4 of~\cite{bagettr}). 
However both the exact statement of  that lemma and its proof are flawed. 
Our proof corrects the argument, making use of model-theoretic techniques
to prove the bounded base lemma.
It then follows the rest of 
the argument in~\cite{bagettr} to show not only Datalog-rewritability, 
but rewritability into a Datalog program comprised of frontier-guarded rules (defined
below).

\myparagraph{The chase}
To prove results about certain answers, we will need to make use of the standard ``Chase construction'' 
for TGDs (see, e.g.~\cite{FKMP05}): given a structure $\frakA$ for signature $\sigma$ 
and a finite set of TGDs $\Sigma$, the chase construction produces 
a structure $\chase_\Sigma(\frakA)$ with the following properties:
\begin{itemize}
\item $\chase_\Sigma(\frakA)$ satisfies $\Sigma$ and $\frakA \subinst \chase_\Sigma(\frakA)$.
\item for any  boolean conjunctive query $Q$ with constants from $\frakA$,
$Q$ is  satisfied in $\frakB$ exactly when it is implied by $\Sigma$ and the facts of $\frakA$. 
\end{itemize}

$\chase_\Sigma(\frakA)$ is formed just by  repeatedly throwing in facts
using fresh elements to witness the heads of  unsatisfied TGDs. 
There are several variations of the chase \cite{FKMP05,onet}, but we describe a construction
that will suffice for our purpose.

$\chase_\Sigma(\frakA)$ is the union of structures 
$\frakB_j$ formed inductively. 
In the base case,  $\frakB_0=\frakA$, while
in the inductive case  $\frakB_{j+1}$ is formed from $\frakB_j$ as follows:
for every $\sigma \in \Sigma$
   $$
      \forall \vec x \, \left( \phi(\vec x) \rightarrow \exists \vec y \, \bigwedge_i A_i(\vec x, \vec y) \right)
      $$
 for every homomorphism $h$ of $\phi$ into $\frakB_j$, add facts $A_i(h(\vec x), \vec y_0)$
to $\frakB_j$, where $\vec y_0$ are values disjoint from $\adom(\frakB_j)$, any
constants of $\Sigma$, and the values used in any other $\sigma, h$ for $\frakB_j$.

Several of the arguments below will involve showing that $Q$ is certain with respect
to $\Sigma$ and $\frakA$ by arguing that $Q$  must  hold in $\chase_\Sigma(\frakA)$.

We will need an additional observation about the chase with
Frontier-Guarded TGDs, which is that the chase has a tree-like structure.
This is well-known \cite{bagetconf}, but it will be useful to state it in terms of our
notion of squid-extension from earlier in the paper.  
\begin{lemma} \label{lem:chasesquid}
If $\Sigma$ consists of frontier-guarded TGDs, then $\chase_\Sigma(\frakA)$ is
a squid-extension of $\frakA$. 
\end{lemma}
\begin{proof}
Letting  $\frakB=\chase_\Sigma(\frakA)$ 
recall that we must show that 
\begin{enumerate}
\item[(i)] every set of elements from the active domain of $\frakA$
      that is guarded in $\frakB$ is already guarded in $\frakA$; and
\item[(ii)] $\frakB \ominus \frakA$ is a union of 
tentacles $\frakB_X$ for $X$ a guarded subset of $\frakA$ such that
for distinct $X$ and $X'$, $\frakB_X$ and $\frakB_{X'}$ 
overlap
in their active domains only in $\adom(\frakA)\cup C$, and finally
      $(\adom(\frakB_X)\cap \adom(\frakA))\setminus C \subseteq X$, where $C$ is
      the set of elements of $\frakA$ named by a constant symbol.
\end{enumerate}

As we generate $\frakB=\chase_\Sigma(\frakA)$ we build the set of tentacles $\frakB_X$ for
each guarded set $X$ in $\frakA$, inductively preserving
the properties above.  Initially $\frakB_X$  contains every fact in $\frakA$
that is guarded by $X$. Clearly, both properties hold.

Recall that the chase is formed as the union of $\frakB_j$, where
$\frakB_{j+1}$ is formed inductively from $\frakB_{j}$ by firing  rules
$\sigma \in \Sigma$ based on a homomorphism $h$ of the body of $\sigma$ into the  structure $\frakB_j$
built so far, generating facts $G$ that are added
to $\frakB_j$. Let $F$ be the image of a guard atom for $\sigma$ under  $h$.
If $F$ is contained in $\frakA$, there is nothing to be done to preserve
the invariants.
If $F$ is not contained in $\frakA$, 
then by the second inductive invariant, $F$ is associated with a  $\frakB_X$ for some $X$ that is 
guarded in $\frakA$.
We add $G$ to $\frakB_X$. 

We show that the inductive invariants are preserved. Clearly $\frakB_{j+1} \ominus \frakA$
is a union of tentacles, since we added $G$ to exactly one tentacle.
Let us consider the first property.
Suppose a set $a_1 \ldots a_k$ of elements of $\frakA$ is guarded by $G$. Then
$a_1 \ldots a_k$ must correspond to exported variables of the rule; that is, none
of them could have been generated as a fresh value in the creation of $G$. Thus
they must be guarded by $X$.

For the second property, 
any new elements added to $\adom(\frakB_X)$ must be disjoint from those in $\adom(\frakB_{X'})$,
and any fact is added to a unique $\frakB_X$. Finally any element added 
to $[\adom(\frakB_X) \cap \adom(\frakA)] \setminus C$  must be contained in the guard atom  $G$, and
by induction this is contained in $X$.
\end{proof}

\myparagraph{Rewriting the certain answers of atomic queries over guarded TGDs}
We start with a result that gives the intuition for how this rewriting works:

\begin{theorem} \label{thm:gfatomicrewrite}
For every set $\Sigma$ of guarded TGDs, and for every
atomic conjunctive query $Q(\textbf{x})$,
one can effectively find a Guarded Datalog program $P$ such
that the output  of $P$ on any structure $\frakA$ is the same as the
certain answers to $Q$ on $\frakA$.
\end{theorem}

Note that entailment here, and throughout the section,
can be interpreted either in
  the classical sense or in the finite sense, since we have the finite
  model property.  Indeed, in our proofs, we use constructions
  that make use of infinite structures, but the conclusion holds in
  the finite.

A \emph{full TGD} is a TGD with no existentials in the head.
The idea behind the proof the theorem will be that we take all full guarded TGDs that
are consequences of $\Sigma$, and turn them into Datalog
rules.  We will show that the full guarded TGDs are sufficient to capture the certain
answers.

We say that a structure $\frakA$ is \emph{fact-saturated} (with respect to $\Sigma$) 
if no new fact over the active domain of $\frakA$ plus the elements named
by constant symbols is entailed by the facts of $\frakA$ 
together with $\Sigma$.  

\begin{lemma} \label{lem:boundedbaseguardedatomic}
For $\Sigma$ a set of guarded TGDs, 
  if  a structure $\frakA$ is not fact-saturated with respect to $\Sigma$, then there is a guarded
  subset $X$ of the domain of $\frakA$ such that the induced substructure
  $\frakA_X$   is not fact-saturated with respect to $\Sigma$.
\end{lemma}
\begin{proof}
We prove the contrapositive.
Assume that every
induced substructure $\frakA_X$, for $X$ a guarded subset, is fact-saturated with
respect to $\Sigma$.
Let $\frakB$ be constructed
from $\frakA$ by chasing each $\frakA_X$ with $\Sigma$ independently
and taking the union of the results: that is $\frakB=\bigcup_{X \mbox{ guarded}} \chase_\Sigma(\frakA_X)$.
Recalling that the chase of $\frakA_X$ only satisfies
facts over $\frakA_X$ that are entailed, we see that $\frakB$ does not satisfy any
new facts over the domain of $\frakA$.

We claim that $\frakB$ satisfies every sentence in $\Sigma$.
Consider a dependency $\sigma$ in $\Sigma$ of the form
\[
\forall \vec x \, \left( \phi(\vec x) \rightarrow \exists \vec y \, \rho(\vec x, \vec y) \right)
\]
and a binding of variables $\vec x$ into $\vec b \in \frakB$ such
that the corresponding facts $\phi(\vec b)$ hold in $\frakB$.
Note that since $\sigma$ is a guarded TGD, $\vec b$ is guarded.
If $\vec b$ contains only constants and elements of $\frakA$, then each fact in $\phi(\vec b)$
must be in $\frakA$. Hence $\phi(\vec b)$ is in $\frakA_X$ and we are
done, since $\frakA_X$ satisfies $\Sigma$.
Consider any non-constant element $b_i$ outside of $\frakA$.
If any such element exists, then the guard fact for $\vec b$ must have been generated in the chase
process for some $\frakA_{X_0}$, hence every non-constant element $b_i$ was
generated in $\frakA_{X_0}$, and every fact in $\phi(\vec b)$ involving
such an element must be in $\frakA_{X_0}$. Since every other fact is  in
$\frakA$, hence in $\frakA_{X_0}$, we have  $\phi(\vec b)$ is contained
in $\frakA_{X_0}$ as before, and so we are done because $\Sigma$ holds in $\frakA_{X_0}$.

Thus we have a structure satisfying $\Sigma$, containing $\frakA$,  and containing no new
facts over the elements of $\frakA$ and the constants. Therefore $\frakA$
must be fact-saturated.
\end{proof}

We are now ready to give the proof of Theorem \ref{thm:gfatomicrewrite}:

\begin{proof}[Proof of Theorem \ref{thm:gfatomicrewrite}]
A \emph{derived full guarded TGD} for $\Sigma$ is a  full guarded TGD that
is entailed by $\Sigma$ and which has a single atom
in the head.
We let $\Sigma_{\fullguarded}$ be all the derived full guarded TGDs.
Note that  once we fix the signature, we fix the maximal number of atoms
in the body of a guarded TGD, assuming that  atoms that
are redundant are eliminated. Thus once we fix both the constants
and the relations in the signatures, we we fix the number
of full guarded TGDs with a single atom in the head, up to renaming
of variables and elimination of redundant atoms.
Thus the number derived full guarded TGDs  in a fixed signature, up 
to renaming and elimination of redundant atoms, 
is  finite.

Lemma \ref{lem:boundedbaseguardedatomic} implies that:

\medskip

For every
$\frakA$ and atomic query $Q=\goal(\vec x)$,  the certain answers  of $Q$ over $\frakA$
with respect to $\Sigma$ are the same as the $Q$-facts entailed
by $\frakA$ and $\Sigma_{\fullguarded}$.

\medskip

The full TGDs of $\Sigma_{\fullguarded}$  are not quite Guarded Datalog.
Guarded Datalog requires us to distinguish extensional and intensional relations, and
requires that atoms over extensional relations do not occur as consequences within rules.
We turn  $\Sigma_{\fullguarded}$ into a Guarded  Datalog
program by replacing
each relation $R$ in $\Sigma_{\fullguarded}$ by a copy $R'$. Thus
a full TGD:
\[
\forall \vec x \vec y \, \left( R(\vec x, \vec y) \ldots \rightarrow  S(\vec x) \right)
\]
is transformed to the Datalog rule:
\[
S'(\vec x) \datalogarrow \exists \vec y ~ R'(\vec x, \vec y) \ldots 
\]

In addition we add rules:
\[
R(\vec x) \datalogarrow R'(\vec x)
\]
Finally, we let $\goal'$ be the goal predicate. It is easy
to see that this Datalog program computes a fact $\goal'(\vec a)$ over $\frakA$
exactly when $\goal(\vec a)$ is entailed by $\Sigma_{\fullguarded}$ over $\frakA$.
\end{proof}

\myparagraph{General conjunctive queries and Guarded TGDs}
We now extend the result to general conjunctive queries. The conference
paper \cite{mfcs14} claimed that 
the certain answers of an arbitrary answer-guarded CQ $Q$ are expressible
in Guarded Datalog. However this is easily seen to be false: indeed even with no constraints
we still need to express that
$Q$ holds in $\frakA$, which is  expressible in Guarded Datalog only if $Q$ is equivalent
to a $\gf$ formula. Thus for any CQ $Q$ that is not in $\gf$, such
as $\exists x y z ~ R(x,y) \wedge R(y,z)$, the certain
answers with respect to the empty set of constraints
are not rewritable in Guarded Datalog.

We thus need to move to a slight extension of Guarded Datalog that allows non-guarded rules at top-level.
We consider Datalog programs where the special relation $\goal$ does not occur in the body
of any rule. Every Datalog program can be rewritten this way.
A \emph{goal rule} in such a Datalog program is one that has the relation $\goal$ in the head.
A Datalog program is \emph{internally-guarded} if for every rule that is not a goal rule, the body has an atom over the input signature that  guards
each variable. That is, internally-guarded Datalog weakens Guarded
Datalog by making an exception for the goal rule.

Recall that a conjunctive query
is answer-guarded if it includes an atomic formula that guards all free variables.
In particular all Boolean conjunctive queries are answer-guarded.

Our goal is the following result.

\begin{theorem} \label{thm:guardedrewrite}
For every set $\Sigma$ of guarded TGDs, and for every
conjunctive query $Q$, one can effectively find an internally-guarded Datalog program $P$ such
that
 on any structure $\frakA$ and binding $\vec c$ for the free
variables of $Q$ in $\frakA$,  $\vec c$ belongs to the output of $P$ on $\frakA$ exactly
when  $\frakA \wedge \Sigma \models Q(\vec c)$.
\end{theorem}

In the proof we will make use of the same construction as in
the case where $Q$ consists of a single atom: given $\frakA$, we take each guarded set $X$ of $\frakA$, and
let $\frakB=\bigcup_X \chase_\Sigma(\frakA_X)$. In the previous proof we
showed that  $\frakB$ satisfies the constraints $\Sigma$.
We note further:

\medskip

The sets $\chase_\Sigma(\frakA_X) \ominus \frakA$ as $X$ ranges over guarded subset
of $\frakA$, form tentacles witnessing that  $\frakB$ is a squid-extension
of $\frakA$.

\medskip

Clearly the active domains of these sets  overlap only in $\frakA$,
and $\frakB \ominus \frakA$ is their union.
From Lemma \ref{lem:chasesquid} we see
that each $\chase_\Sigma(\frakA_X)$ has no new guarded sets which contain only elements in $\frakA$.

We now turn to the construction of the Datalog program that
witnesses Theorem \ref{thm:guardedrewrite}.
The idea  will be to add new relations for certain
guarded queries
derived from $Q$, along with full guarded TGDs that capture their semantics.
For each 
query $q$ of size at most that of $Q$, let  $R_q$ be a new relation symbol.

For a query $q$ with variables $\vec x$ free and
at most $k$ variables, a \emph{guarded query generation rule} for $q$ is  a full TGD of the form:
\[
\forall \vec x \vec y ~ (A_0(\vec x, \vec y) \wedge \bigwedge_{i=1 \ldots n}  A_i(\vec x, \vec y)  \rightarrow R_q(\vec x))
\]
where each $A_i$ is an  atom over the signature of $\Sigma$ whose free variables are contained
in the atom $A_0(\vec x, \vec y)$, and
the corresponding TGD
\[
\forall \vec x \vec y ~ (A_0 \wedge \bigwedge_i  A_i  \rightarrow q(\vec x))
\]
is a consequence of $\Sigma$.
Notice that:
\begin{itemize}
\item guarded query generation rules are guarded TGDs
\item  there are only finitely many guarded query generation rules (up to logical
equivalence) since there are only finitely many guarded conjunctions
\item determining whether a TGD is a guarded query generation rule 
can be determined effectively, using the  decidability of $\gnfo$
\end{itemize}

\begin{proof}[Proof of Theorem \ref{thm:guardedrewrite}]
Let $k$ be the maximal number of variables in the body of  a rule of $\Sigma$.
Consider the  signature with intensional relations $R_q$ for every query
$q$ with at most $k$ variables. Consider the set of full TGDs  $P_{Q}$ consisting of:
\begin{itemize}
\item All derived full guarded TGDs  (over the original signature)
\item All guarded query generation rules 
\item As goal rules, all TGDs of the form
\[
\bigwedge_j R_{q_j}(\vec x, \vec y) \rightarrow \goal(\vec x)
\]
such that each $q_j$ is a CQ with at most $k$ variables,  
$ \bigwedge_j q_j$ entails $Q(\vec x)$, and
the number of variables in the rule is at most $k$.
\end{itemize}

We can compute the last set of TGDs using the decidability of conjunctive
query containment.

We claim that for any $\vec c \in \frakA$, $\goal(\vec c)$ is entailed by $\frakA \wedge P_{Q}$
if and only if  $Q$ is entailed by $\frakA \wedge \Sigma$.

In one direction, suppose $\goal(\vec c)$ is entailed by $\frakA \wedge P_{Q}$.
Then  there is a single goal rule $\sigma$ of form
\[
\bigwedge_j R_{q_j} \rightarrow \goal(\vec x)
\]
that  derives $\goal(\vec c)$, based
on previously derived facts $P^-$.
Note that facts over the auxiliary relations $R_{q}$ can only be generated from
guarded  query generation rules.
Thus $P^-$ consists of
facts $R_{q_j}(\vec c_j)$ which are each generated by applying a
guarded query generation rule to a  set of facts $F_j$ 
where the $F_j$ include a guard fact $G_j$ over  $\frakA$.
We will be able to  conclude that $\goal(\vec c)$ is derived from $\frakA \wedge \Sigma$, using the definition of the guarded query generation rules and the goal rules,
assuming that we can conclude that 
each set of  facts $F_j$ is derived from $\frakA \wedge \Sigma$.
But  each fact in $F_j$  must  have been generated from $\frakA$ in $P_Q$
by applying derived guarded rules.
Thus by definition of these rules, each of them are a consequence of $\frakA \wedge \Sigma$.

In summary,  all of the facts that lead to the firing of $\sigma$
are consequences of $\frakA$ and $\Sigma$.

We now turn to the other direction, showing
that if $Q(\vec c)$ is entailed by $\frakA \wedge \Sigma$, then 
$\goal(\vec c)$ is entailed when $P_Q$ is applied to $\frakA$.
We know that $Q(\vec c)$ holds in $\frakB=\bigcup_X \chase_\Sigma(\frakA_X)$
defined above.  We thus have a homomorphism
$h$ from $Q(\vec c)$ into $\frakB$.
Let $h_Q$ be the image of the atoms in $Q$ under $h$.
 Then $h_Q=  \bigcup_{i \leq n} F_i$, where $F_i$
lies in  $\chase_\Sigma(\frakA_{G_i})$
for a guarded set $G_i$ in $\frakA$.

Let CQ $q_j(\vec c_j)$ be obtained from $F_j$ by turning each element outside of 
$\frakA$ into an existentially quantified variable and keeping
the elements within $\frakA$ as constants.
 By the definition of
$\chase_\Sigma(\frakA_{G_j})$,  we have that $q_j$ is entailed by the facts
over the guarded set $G_j$ using $\Sigma$. 
Thus we have a corresponding guarded query generation rule with $R_{q_j}$ in the head.
 By our prior results on the atomic case, each fact in $G_j$  is entailed
 by $P_{Q}$.
Combining these last two statements we see that $R_{q_j}(\vec c_j)$ is entailed
from $\frakA$ and $P_{Q}$.

Since there is a   homomorphism  of $Q$ to the union of atoms in each $q_j$, we  see
that the conjunction of the $q_j$ entails $Q$. Thus we have
a  corresponding goal rule  in $P_{Q}$:
\[
 \bigwedge_j R_{q_j} \rightarrow \goal(\vec x)
\]
Since facts matching the hypotheses of this rule are derived from $P_{Q}$
on $\frakA$, firing
this last rule allows us to conclude that $\goal(\vec c)$ is entailed from
$P_Q$ on $\frakA$ as required.
\end{proof}

\myparagraph{Frontier-guarded TGDs}
We now generalize the result about rewriting certain answers to frontier-guarded TGDs.
By a frontier-guarded rule in a Datalog program we mean a rule
whose body contains an atomic formula that guards all variables
that appear also in the head.
A \emph{Frontier-guarded Datalog program} is a Datalog program in
which each rule is frontier-guarded.

\begin{theorem} \label{thm:fgcert}
For every set $\Sigma$ of frontier-guarded TGDs, and for every
answer-guarded conjunctive query $Q(\textbf{x})$, 
one can effectively find a frontier-guarded Datalog program $P$ such
that the output  of $P$ on any structure $\frakA$ is the same as the 
certain answers to $Q$ on $\frakA$.
\end{theorem}

  We can assume without loss of generality that $Q$ is an atomic query
  (by extending $\Sigma$ with an extra ``answer rule'' containing the
  query. This rule is frontier-guarded because $Q$ is answer-guarded).
We will also assume that for each relation $R$ of arity $n$ and each
subset $S=i_1 \ldots i_k$ of $\{1 \ldots n\}$  there is a new ``guard extension predicate''
$R_S$ of arity $k$, and  dependencies:
\[
R(x_1 \ldots x_n) \rightarrow  R_S(x_{i_1} \ldots x_{i_k})
\]
and
\[
R_S(x_{i_1} \ldots x_{i_k}) \rightarrow \exists \vec x ~ R(\vec x)
\]
where $\vec x$ denotes $x_j$ for $j \notin S$.

We can obviously add such dependencies, and a rewriting using these predicates
can be replaced with a rewriting using the original predicates.
Thus for every guarded set  in the original vocabulary, we have an atomic predicate
that holds of exactly those elements in the vocabulary with guarded extensions.


We will create new predicate symbols for certain queries, as we did in
Theorem \ref{thm:guardedrewrite}. 
  Let $k$ be the maximal number of variables  in a TGD of $\Sigma$.
For an answer-guarded 
  conjunctive query $q(x_1 \ldots x_j)$ in the guard extension vocabulary above, 
let
$R_q(x_1 \ldots x_j)$ be a relation symbol, a ``query extension predicate''. For any number $k$, 
let $\fgtgd_k$ be all the frontier-guarded TGDs in the signature extending
$\Sigma$ with each $R_S$ and each $R_q$ for each answer-guarded $q$ in the extension vocabularies
above, with the TGD having at most $k$-variables.
Let $\Sigma'_k$ be all TGDs in $\fgtgd_k$ that are consequences of
\[
  \Sigma \cup \{ \forall \vec x ~ R_q \leftrightarrow q  \mid  
                 q \text{ answer-guarded CQ with $\leq k$ variables} \}.
\]

For a structure $\frakA$, let $C_\frakA$ be the set of elements of $\frakA$ 
named by  constant symbols.

We now convert the full TGDs in $\Sigma'_k$ to a Datalog program, in the same
way as we did in  Theorem \ref{thm:gfatomicrewrite} and Theorem \ref{thm:guardedrewrite}.
That is,
We let $P_{\Sigma,Q}$ be a Datalog program with all  full rules in $\Sigma'_k$, over
a copy of the signature of $\Sigma$, along with the additional extension predicates
$R_S$ and $R_q$ with all  predicates being intensional. In addition
we have rules stating that every relation of $\Sigma$ is contained in its copy.
We will show that $P_{\Sigma, Q}$ is the desired rewriting.
Since running $P_{\Sigma,Q}$ is the same as running all the full rules in $\Sigma'_k$,
up to the difference between a fact and its copy,
this will involve arguing that if we start with a structure $\frakA$ and add all the facts
produced by the full rules $\Sigma'_k$, then we get a structure that is fact-saturated
with respect to $\Sigma$. We
will thus need some characterizations of when a structure is fact-saturated.
We start with a lemma that holds for arbitrary frontier-guarded TGDs.

 We say that $\frakA$ is
  \emph{guardedly fact-saturated} (with respect to a set of TGDs $\Sigma$) if every
possible
  fact
  over $\adom(\frakA)\cup C_\frakA$ entailed by the facts of $\frakA$ together with
  $\Sigma$, \emph{such that the values occurring in the fact form a  guarded set in
  $\frakA$}, belongs to $\frakA$.
 In the absence of constants, 
guardedly fact-saturated means that the structure captures every entailed fact over $\adom(\frakA)$ guarded by
an existing ground atomic formula of $\frakA$.

We then show:
\begin{lemma} \label{lem:boundedbaseguardedly}
If  structure is guardedly fact-saturated with respect to a set of frontier-guarded
TGDs $\Sigma$, then it is fact-saturated with respect to $\Sigma$.
\end{lemma}
Note the difference from Lemma \ref{lem:boundedbaseguardedatomic}. There
the sufficient condition for $\frakA$ to be saturated was that  $\frakA$ was closed under 
\emph{applying a saturation procedure to each guarded set in isolation}. Here our sufficient
condition is that saturating $\frakA$ in its entirety does not miss  any fact guarded over
$\frakA$.

\begin{proof}
Assume $\frakA$ is guardedly fact-saturated.
We consider
$\chase_\Sigma(\frakA)$, and show that any fact n it whose elements are either in $\frakA$ or are named by constants
must already be in $\frakA$.  This is intuitive when we consider that $\chase_\Sigma(\frakA)$ is a squid-extension,
with every fact in the tentacles generated by a guarded set in $\frakA$.

Formally, we prove the following stronger claim: for every fact $F$ in  $\chase_\Sigma(\frakA)$,
the set of elements in $F$ within $\adom(\frakA)$ is guarded in $\frakA$.
If the claim is true, then a fact that used only elements in $\adom(\frakA)$ union constants,
must be guarded, and then since $\frakA$ is  guardedly fact-saturated such a fact must already be in $\frakA$.
The claim  is proven by induction on the generation of $\chase_\Sigma(\frakA)$.
Considering an  application of a rule $\sigma$ that produced a fact $F$, there is a guard
atom matching the body of the frontier-guarded of $\sigma$, produced at an earlier stage and
 containing all the elements of $F$ that are in $\adom(\frakA)$. Now by induction we are done.
\end{proof}

  We now claim the following ``bounded base lemma'' which differs
from Lemma \ref{lem:boundedbaseguardedatomic}  and Lemma \ref{lem:boundedbaseguardedly}
by considering small subsets, but not guarded ones:
\begin{lemma} \label{lem:boundedbasefgtgd}
Letting $k$ be the maximal number of variables in a TGD of $\Sigma$,
and let  $\frakA$ be  a structure such that for each subset $X$ of the domain
of $\frakA$ with $|X| \leq k$, the induced substructure $\frakA_X$ is fact-saturated with respect to $\Sigma'_k$.
Then $\frakA$ is fact-saturated with respect to $\Sigma'_k$.
\end{lemma}

A  lemma similar to Lemma \ref{lem:boundedbasefgtgd}
 occurs in Marnette's unpublished work  \cite{marnette} (Marnette's ``bounded
depth property'').

\begin{proof}
Suppose that every
substructure $\frakA_X$ of $\frakA$ with $|X|\leq k$ is fact-saturated. Let 
$\chase_{\Sigma'_k}(\frakA_X)$ 
be the result of the chase with $\Sigma'_k$ on $\frakA_X$.
Note that by the second property of the chase mentioned at the beginning
of the section,  all the facts over $\frakA_X$ in $\chase_{\Sigma'_k}(\frakA_X)$
are entailed by $\Sigma$ and $\frakA_X$,
Since  $\frakA_X$ is fact-saturated, we deduce that $\chase_{\Sigma'_k}(\frakA_X)$ does not
contain any additional facts over the set $X$ plus the set of elements
named by constant symbols.  We now define $\frakB$ to be the
union of all these $\chase_{\Sigma'_k}(\frakA_X)$.
By
construction, $\frakB$ extends $\frakA$ and contains no new guarded facts over
$\adom(\frakA)$ and the elements named by constant symbols. 
Further, note that  $\adom(\chase_{\Sigma'_k}(\frakA_X))$
for different $X$'s overlap only on $\adom(\frakA)$ and
the elements named by constant symbols. Using Lemma \ref{lem:chasesquid}
we can see that $\frakB$ represents
a squid-extension of $\frakA$, with each tentacle contained
in one of the $\chase_{\Sigma'_k}(\frakA_X)$.

We  will show that $\frakB\models\Sigma'_k$. If we can show this, 
 it would follow that any fact over $\frakA$ entailed by $\Sigma'_k$ must already
lie in $\frakB$. And since $\frakB$ is the union of structures fact-saturated over $\frakA$, any
such fact must lie in $\frakA$.
So we would have proven that $\frakA$ is  fact-saturated, as required.

Consider a
frontier-guarded TGD $\sigma$ in $\Sigma'_k$ of the form
$\forall\textbf{x}(\phi(\textbf{x})\to
\exists\textbf{y}\psi(\textbf{x,y}))$ that is not satisfied. 
and a map $h:\{\textbf{x}\}\to \adom(\frakB)$.
We need to show that $h$ extends to a homomorphism of $\psi$.

Let $\vec n_0$ be the $h$-image of the frontier variables of $\phi$, and $\vec H$ be the 
entire $h$-image.
$\vec H$ decomposes into sets $\vec H_i$ in the different tentacles $T_i$.
The set $\vec n_0$ is a guarded set, so it must lie in one tentacle $T_0$, which
we call the ``main tentacle'', while the other $T_i$ are denoted as ``side tentacles''.

Fix a $\vec H_i$ lying in  side tentacles $T_i$  and let $G_i$ be a guarded set that forms the intersection
of $T_i$ and $\frakA$.  Let $q_i$ be a CQ formed from taking the image under $h$ of all atoms over $\vec H_i$
in $\phi$, with elements of $\vec H_i$ transformed into variables, existentially quantifying over  any
variables whose $h$-image does not lie in $G_i$. We also add on to $q_i$ an atom corresponding to the guard atom of $G_i$, existentially
quantifying away variables corresponding to element of $G_i$ not in $\vec H_i$.
Thus $q_i$ is an answer-guarded CQ with at most $k$ variables that holds of the elements $\vec K_i=\vec H_i \cap G_i$.
Since  these elements lie in tentacle $T_i$, which in turn lies inside $\chase_{\Sigma'_k}(\frakA_X)$ for some $X$, and this
latter structure satisfies $\Sigma'_k$, 
we know that $R_{q_i}$ must hold of $\vec K_i$ in $\frakB$.
Recalling that $\frakA$ is fact-saturated for small sets and that $\vec K_i$ is
a small subset of  $\frakA$, 
we see that $R_{q_i}$ must have already held of $\vec K_i$ in $\frakA$.

Let $G_0$ be a guarded set consisting of  the intersection  of the elements in the main tentacle $T_0$ and $\adom(\frakA)$.
Let $q^*_0$ be a Boolean CQ with variables for all elements of the image $\vec H$.
We will have  atoms
corresponding to each fact in the guard extension signature over $\vec h$ that lie in the image of $h$,
and the free variables will be those corresponding to elements in $\vec H$ intersected with $G_0$.
$q^*_0$ has at most $k$ variables, and it is answer-guarded, since the elements of $G_0$ will be guarded
by a guard-extension predicate.
Thus we have a query extension predicate $R_{q^*_0}$.

Let $q_0$ be a CQ with variables for all elements that lie in the intersection of $\vec H$ and the domain
of $\frakA$.
$q_0$ has atoms corresponding to facts over this set in $\frakB$ and 
also facts $R_{q_i}$ that hold on atoms in the side tentacles. The free variables, as in $q^*_0$ will
be the variables corresponding to elements of $G_0$.
$q_0$ is also answer-guarded, although it is not in the guard extension vocabulary.
The following dependency is a consequence of $\Sigma'_k$:
\[
q_0(\vec x_0) \rightarrow R_{q^*_0}(\vec x_0)
\]
Letting $\vec g_0$ be a binding of the variables corresponding to  $G_0$ with the associated elements,
we have that $R_{q^*_0}(\vec g_0)$ is entailed by $\Sigma'_k$ and $\frakA$.
Again, appealing to the fact that small subset of $\frakA$ are fact-saturated for $\Sigma'_k$, keeping
in mind that the intersection of $\vec H$ and the domain of $\frakA$ is small, we conclude that
$R_{q^*_0}$ holds of $\vec g_0$ in $\frakA$.

Consider the subquery $\phi_0$ of $\phi$ formed by removing
all atoms that are mapped by $h$ into $T_0$ 
adding the  fact $R_{q^*_0}$ on the variables of $\phi$ mapped by $h$ into $G_0$.
Letting $h'$ be the restriction of $h$ to these variables, we see that $h'$ is a homomorphism
of  $\phi_0$.
Letting $\sigma'$ be the analogous modification of $\sigma$:
\[
\forall\textbf{x}(\phi_0(\textbf{x})\to
\exists\textbf{y}\psi(\textbf{x,y}))
\] 
Then $\sigma'$ is entailed by $\Sigma'_k$. Since $T_0$ is contained in
some $\chase_{\Sigma'_k}(\frakA_X)$ that satisfies $\Sigma'_k$, $h'$ extends
to a homomorphism of $\psi$.
This clearly serves as an extension of $h$, and thus we have completed the proof of 
 Lemma \ref{lem:boundedbasefgtgd}.
\end{proof}

We are now ready to prove Theorem \ref{thm:fgcert}.

\begin{proof}[Proof of Theorem \ref{thm:fgcert}]
To show that $P_{\Sigma,Q}$ is the desire rewriting, we
start with a structure $\frakA$ and let $\frakA^+$ be the result
of running $P_{\Sigma,Q}$ on it. Since it is clear that running
$P_{\Sigma,Q}$ does not produce facts that are not entailed, it is enough to show that
if $Q$ is entailed by $\frakA_0$ and $\Sigma$, the copy of $Q$ (over the intentional
signature of $P_{\Sigma,Q}$) holds
in $\frakA^+$. Since $P_{\Sigma, Q}$ is, up to the distinction between a relation and its copy,
the same as the full rules in $\Sigma'_k$, this boils down to showing that saturating
with the full rules of $\Sigma'_k$ gives a structure fact-saturated for $\Sigma$.

To see this, let $\frakA^+$ be formed by closing $\frakA$ under all full rules in $\Sigma'_{k}$.
We claim $\frakA^+$ is fact-saturated for $\Sigma$. By
 Lemma \ref{lem:boundedbasefgtgd} it suffices to show
that given a subset $B$ of size at most $k$
of size at most $k$, the restriction of  $\frakA^+$ to $B$ is fact-saturated  for $\Sigma$.
Clearly it suffices to show that this structure is fact-saturated for $\Sigma'_k$.

By Lemma \ref{lem:boundedbaseguardedly} 
(which holds for all frontier-guarded
TGDs, and hence in particular to $\Sigma'_k$), 
it is enough to show that $B$ contains every fact  entailed by $\Sigma'_k$ 
that is over a set guarded in $\frakA^+$.
Let $\{B_1(\vec c_1) \ldots B_j(\vec c_j)\}$ be all the facts in the initial structure $\frakA$ over $B$.
Consider a fact $F(\vec c)$ with $\vec c$ contained in a guarded subset of $B$ such that
$F(\vec c)$ is entailed by $B$ under $\Sigma'_k$ but is not in $B$. But
then the rule $B_1(\vec x_1) \ldots B_j(\vec x_j) \rightarrow F(\vec x)$ is in
$\Sigma'_{k}$, and it is a full rule. The associated Datalog rule, formed by just switching to the copy predicates used
in $P_{\Sigma, Q}$, is thus
 in $P_{\Sigma, Q}$. Thus applying this rule we get that $F(\vec c)$ holds in $\frakA^+$ as required.
This completes the proof of Theorem \ref{thm:fgcert}.
\end{proof}

\myparagraph{Consequences for deciding FO-rewritability}
In \cite{bbo}, a fragment of Datalog, denoted \emph{GN-Datalog} was defined, and it was
shown that for this fragment one can decide whether a query is equivalent to a first-order
query (equivalently, as shown in \cite{bbo}, to some query obtained by unfolding the Datalog rules
a finite number of times).
Since GN-Datalog contains frontier-guarded Datalog, we can couple
the decision procedure from \cite{bbo} with the algorithm in Theorem \ref{thm:fgcert}
to obtain decidability. In fact, we can obtain the result for general conjunctive queries, not
just answer-guarded ones:

\begin{corollary} \label{cor:decidefo}
FO-rewritability of conjunctive queries $Q$ under sets of
frontier-guarded TGDs $\Sigma$ is decidable.
\end{corollary}

\begin{proof}
In the case where $Q$ is a boolean conjunctive query, we use the technique above:
obtain a frontier-guarded Datalog rewriting and then checking whether it is equivalent
to a first-order formula using  the result of \cite{bbo}.

Now consider the case where $Q$ is a general conjunctive query.
We can form a boolean CQ $Q^*$ by changing
the free-variables $x_1 \ldots x_n$ of $Q$ to constants $c_1 \ldots c_n$. Theorem \ref{thm:fgcert} implies
that we can decide whether the certain answers to $Q^*$ with respect to $\Sigma$ are first-order definable.
But the certain answers of $Q^*$ with respect to $\Sigma$ are first-order definable
if and only if the certain answers to $Q$ with respect to $\Sigma$ are first-order definable:
we can change a first-order definition of one to a first-order definition of the other
by just replacing constants with free variables or vice versa.
\end{proof}

\section{Related Work and Conclusions} \label{sec:conc}

We have investigated various problems that involve rewriting of $\gnfo$ formulas in different contexts,
building on the decidability results for $\gnfo$
established in \cite{BtCS15jacm}, and the complexity results for open-
and closed-world querying established in \cite{bbo}.

Although we did not discuss the exact
complexity of the decision problem for FO-rewritability of certain answers under frontier-guarded
TGDs,   we believe that an elementary bound can be extracted from analysis of
\cite{bbo}. 
Prior to that work, we know of no result on  deciding first-order rewritability in the setting
of general relational languages. However, for description logics, some positive results were obtained
by Bienvenue, Lutz, and Wolter \cite{deciderewrite}.
In \cite{pods13BtCLW}, it was shown that certain answers w.r.t.~a $\gnfo$ sentence
can be expressed in frontier-guarded \emph{disjunctive} Datalog. Unlike our
result for frontier-guarded TGDs, however, this characterization is not known
to imply decidability of first-order rewritability or even Datalog-rewritability.
Weakly-guarded TGDs \cite{CGK08kr} are another member of the Datalog$^{\pm}$ family
that has been shown to have attractive properties for the complexity of open-world query answering. One can show, however, that
they do not share with $\fgtgd$'s the decidability of FO-rewritability.

Here we have considered syntactically capturing restrictions of $\gnf$, and show that the corresponding target classes
for rewritings are natural. 
For description logics, some characterizations with a similar flavor have been proven by Lutz, Piro, and Wolter \cite{dlsem}.
The Unary Negation  Fragment is another fragment of FO containing
many modal and description logics which possesses the
Craig Interpolation Property and (hence) the Projective Beth
Definabiity Property   \cite{tCS11stacs}.
Interpolation and implicit definability have  also been  heavily studied within   the description logic community
\cite{lwinterpol,balderinterpol}.
Unfortunately, having the Beth Definability Property or the
Craig Interpolation Property for a stronger logic does not imply it for a weaker logic, or vice
versa.

Recently, in follow-up work~\cite{csllics14},  tight bounds on the complexity were 
found for a number of problems
considered here, including interpolation and preservation results.


\myparagraph{Acknowledgements.}
This paper is an expanded version of the conference abstract \cite{mfcs14}.
Benedikt was supported by EPSRC grant EP/H017690/1, and
ten Cate was supported by NSF Grants IIS-0905276
  IIS-1217869.
B\'ar\'any's work was done while affiliated with TU Darmstadt.

The authors gratefully acknowledge their debt to Martin Otto for enlightening
discussions. We want to thank
Maarten Marx for
helpful discussions and help in verifying the counterexamples of Section \ref{sec:interpol}. We also thank the anonymous reviewers of the Journal of Symbolic
Logic for their patient reading of the manuscript and helpful corrections.

\bibliographystyle{asl}
\bibliography{gnfo}

\end{document}